\documentclass[letterpaper,journal]{IEEEtran}
\usepackage[backend=bibtex, style=ieee]{biblatex}
\usepackage{amsmath,amsfonts,amsthm}
\usepackage{bbm}
\usepackage{xcolor}
\usepackage{hyperref}
\usepackage{array}
\usepackage{subcaption}
\usepackage{caption}
\usepackage{textcomp}
\usepackage[shortlabels]{enumitem}
\usepackage{stfloats}
\usepackage{url}
\usepackage{verbatim}
\usepackage{graphicx}
\usepackage{balance}
\usepackage{algorithm}
\usepackage{algpseudocode}
\usepackage{stmaryrd}
\usepackage[bibliography=common]{apxproof}
\usepackage{booktabs}
\usepackage{multirow}
\newcommand{\col}{\mathrm{col}}
\newcommand{\Tr}{\mathrm{Tr}}
\DeclareMathOperator*{\argmin}{argmin}
\newtheoremrep{theorem}{Theorem}
\newtheoremrep{corollary}{Corollary}
\newtheorem*{remark}{Remark}
\newtheorem{lemma}{Lemma}

\begin{document}
\title{Signal Subspace Methods Which Are Robust to Impulsive Noise}
\author{Robert L. Bassett \& Micah Y. Oh
\thanks{R. Bassett and M. Oh are with the Naval Postgraduate School's Operations Research department. Both authors acknowledge support from ONR grants N0001421WX00142 and N0001423WX01316}}


\maketitle

\begin{abstract}
We consider the problem of estimating a signal subspace in the presence of interference that contaminates some proportion of the received observations. Our emphasis is on detecting the contaminated observations so that the signal subspace can be estimated with the contaminated observations discarded. To this end, we employ a signal model which explicitly includes an interference term that is distinct from environmental noise. To detect when the interference term is nonzero, we estimate the interference term using an optimization problem with a sparsity-inducing group SLOPE penalty which accounts for simultaneous sparsity across all channels of the multichannel signal. We propose an iterative algorithm which efficiently computes the observations estimated to contain interference. Theoretical support  for the accuracy of our interference estimator is provided by bounding its false discovery rate, the expected proportion of uncontaminated observations among those estimated to be contaminated. Finally, we demonstrate the empirical performance of our contributions in a number of simulated experiments.
\end{abstract}


\section{Introduction}
\IEEEPARstart{E}{stimating} unknown signal or sensor parameters is an essential part of many signal processing problems. When measurements are taken using sensor arrays, the signal subspace can be used to estimate many parameters of practical interest, including signal frequencies, directions of arrival and sensor array geometry. In this paper, we contribute a new technique for estimating a signal subspace in the presence of sparse interference. Throughout, sparsity is measured with respect to time, so a sparse interferer is one which contaminates some bounded proportion of received observations across any subset of the sensor array's channels. An important example of sparse interference is impulsive noise\footnote{also called \emph{impulse noise} in some references}, which refers to high-intensity short-duration signals \cite{henderson1986impulse, zoubir2012robust}. Impulsive noise is commonly found in many application areas, including radar \cite{abramovich1999impulsive}, 
sonar \cite{etter2018underwater}, 
seismology (due to seismic events), 
and telecommunications \cite{middleton1999non}. Since impulsive noise typically does not contaminate all observations, it is an example of sparse interference.


We modify a conventional noise model, in which received observations are a signal with additive noise, by including an additional interference term, and then estimate this term in an optimization problem that features a data fidelity term and a sparsity-inducing regularizer. Once we have identified observations for which the interference term is estimated to be nonzero, we discard those observations and estimate the signal subspace using the remaining observations. In general, this is a difficult problem because identifying a size $k$ subset of observations to discard due to interference is NP-hard \cite{davis1997adaptive}. Sparsity-inducing regularizers can be used to tractably approximate these difficult subset selection problems, and a number of recent results establish rigorous approximation guarantees for a variety of different regularizers \cite{wainwright2019high, tibshirani1996regression, zou2005regularization, yuan2006model, zou2006adaptive, belloni2011square, bogdan2015slope}. Our contributions build on this literature by applying a sparsity-inducing regularizer to signal subspace estimation. By exploiting the unique structure of the design matrix in signal subspace estimation, we are able to develop fast algorithms that compute the estimator much more quickly than general-purpose solvers. In addition to these computational results, we provide theoretical support for the accuracy of the estimator by bounding its false discovery rate, the expected of ratio of incorrectly discarded observations to total discarded observations.


Our approach differs from existing work on signal subspace estimation in a number of ways. Typical models assume that the complex output of an $m$ element array $x(t) \in \mathbb{C}^{m}$, in IQ format, is related to the signal amplitude $s(t) \in \mathbb{C}^{d}$ from $d$ sources via
\begin{equation}\label{easy_model}
x(t) = A_{0}\, s(t) + \epsilon(t),
\end{equation}
where $A_{0} \in \mathbb{C}^{m \times d}$ is a matrix with columns the array response vectors for each of the $d$ sources. The $\epsilon(t) \in \mathbb{C}^{m}$ term is additive noise. In the nonrobust setting, $\epsilon(t)$ is often assumed to be distributed as a mean zero, isotropic complex Gaussian and independent across the time index $t$. The \emph{signal subspace} is the column space of $A_{0}$, denoted $\col(A_{0})$. Denote by $X$ an $m \times n$ matrix consisting of $n$ observations from the model \eqref{easy_model}. If the signal is deterministic but unknown, the signal subspace can be estimated using deterministic maximum likelihood. This leads to the following optimization problem for subspace fitting, first introduced in \cite{schmidt1981signal},
\begin{equation} \label{opt_problem1}
\min_{\substack{ A \in \mathbb{C}^{m \times d}\\ S \in \mathbb{C}^{d \times n}}} \left\|X - A S \right\|_{F}^{2}.
\end{equation}

The authors of \cite{cadzow1988high} observed that substituting other functions of the observations for the $X$ term in \eqref{opt_problem1} results in a framework that generalizes many methods for estimating a signal subspace. For example, the following problem removes the dependence on $n$ in \eqref{opt_problem1} and produces an identical estimate for $\col(A_{0})$,
\begin{equation} \label{opt_problem2}
\min_{\substack{ A \in \mathbb{C}^{m \times d}\\ S \in \mathbb{C}^{d \times m}}} \left\|\hat{\Sigma}_{X}^{1/2} - A S \right\|_{F}^{2},
\end{equation}
where $\hat{\Sigma}_{X}$ is the sample covariance of $X$ and $\hat{\Sigma}_{X}^{1/2}$ a Hermitian square root.

Other methods can be formulated by substituting different functions of the sample covariance matrix $\hat{\Sigma}_{X}$ into the optimization problem \eqref{opt_problem2}. The multidimensional version of the popular MUSIC algorithm \cite{schmidt1981signal, paulraj}, can be written
\begin{equation} \label{opt_problem_music}
\min_{\substack{ A \in \mathbb{C}^{m \times d}\\ S \in \mathbb{C}^{d \times d}}} \left\|\hat{E}_{S} - A S \right\|_{F}^{2},
\end{equation}
where $\hat{E}_{S} \in \mathbb{C}^{m \times d}$ is a matrix of the $d$ largest eigenvectors of $\hat{\Sigma}_{X}$. Because many signal subspace methods depend on $\hat{\Sigma}_{X}$, efforts to robustly estimate a signal subspace primarily focus on robustly estimating the covariance matrix \cite{zoubir2012robust}. Since the covariance matrix is a fundamental quantity in statistics, many robust techniques have been developed and then used to robustify signal subspace methods. Examples include M-estimation \cite{huber1981robust, zoubir2012robust}, sign covariance matrix \cite{oja1997robust, visuri2001subspace}, minimum volume ellipsoid \cite{rousseeuw1984least}, and minimum covariance determinant \cite{rousseeuw1984least}.

Our approach differs from this common paradigm for robustifying signal subspace methods by returning to the probabilistic model \eqref{easy_model}. We include an interference term in \eqref{easy_model} and modify the least-squares estimator \eqref{opt_problem1} accordingly. Similar to the signal of interest $s(t)$ in deterministic maximum likelihood, we assume that this interfering signal is deterministic yet unknown. By assuming that this interference term is sparse across time, we are able to distinguish it from the noise term and signal of interest by forming an optimization problem which modifies \eqref{opt_problem1} to include a sparsity-inducing penalty.

We use the group SLOPE norm as the sparsity-inducing penalty when formulating the estimator. First introduced in \cite{gossmann2015identification} and \cite{brzyski2019group}, the group SLOPE norm encourages simultaneous sparsity across groups of estimated variables. Group SLOPE is an extension of the popular group LASSO penalty \cite{yuan2006model}, which also encourages selected groups of variables to be zero simultaneously. However, group SLOPE has an advantage over group LASSO because it is adaptive, in the sense that the penalty applied to each group depends on its $\ell^{2}$ norm relative to the norms of all other groups, with larger group norms being penalized more heavily. Moreover, the group SLOPE norm is convex, so minimizers of many group SLOPE regularized optimization problems can be computed efficiently. In this paper, we encourage group sparsity across columns of a matrix $\Delta \in \mathbb{C}^{m \times n}$, where $m$ is the number of channels in the sensor array and the $n$ columns of $\Delta$ are each an estimate of an interference term. This column-wise grouping encourages sparsity across all rows of a column of $\Delta$ simultaneously, and a column of $\Delta$ set to zero indicates that no interference is estimated across any channel in that sample.

The rest of this paper proceeds as follows. In the next subsections of this section, we introduce our signal model, a robust estimator for the signal subspace using this model, and summarize our main results. Section \ref{sec:computation} focuses on computing the estimator, where we present an iterative algorithm which exploits the unique structure of sparsity-inducing penalties applied to signal subspace estimation. We also present a convergence result for this algorithm. Section \ref{sec:error} examines the statistical error of our estimator. We provide results for tuning the estimator to accommodate different levels of tolerance for interference and provide a bound on its false discovery rate--the expected proportion of incorrectly identified observations among those estimated to contain interference. In section \ref{sec:experiments}, we conduct a number of simulations to demonstrate the performance of our estimator under a variety of conditions. Section \ref{sec:conclusion} concludes the paper. Proofs which are not included in the main text are relegated to the Supplementary Material.

\textbf{Notation}. Before proceeding, we establish some notation that we will use throughout the paper. We denote matrices by capital letters (both Latin and Greek) and column vectors by lower case letters. We denote by $\|\cdot\|_{F}^{2}$ the squared Frobenius norm and by $\|\cdot\|_{2}^{2}$ the squared Euclidean norm. Given a matrix $X$, we denote the column space of $X$ by $\col(X)$ and the matrix that orthogonally projects onto $\col(X)$ by $P_X$. Denote by $x_{i}$ the $i$th entry of a vector $x$ and, when $x_{i}$ is  complex-valued, $|x|_{i}$ the modulus of its $i$th entry. For a real-valued vector $x$, we denote by $x_{(i)}$ the $i$th order statistic of $x$, so that $(x_{(1)}, x_{(2)}, ... )$ gives the entries of $x$ in nondecreasing order. We denote the set of natural numbers $\{1,\dots,n\}$ by $[n]$. The transpose and conjugate-transpose of a matrix $X$ are denoted $X^{\top}$ and $X^{*}$, respectively. For $X \in \mathbb{C}^{m \times n}$, we denote by $\llbracket X \rrbracket \in \mathbb{R}^{n}$ the vector consisting of the Euclidean norms of each column of $X$. We also denote the $i$th column of $X$ as $X_{i}$. We denote the multivariate normal distribution, with mean $\mu \in \mathbb{R}^{n}$ and symmetric covariance matrix $\Sigma \in \mathbb{R}^{n \times n}$ by $N(\mu, \Sigma)$, and we denote the complex multivariate normal distribution with mean $\mu \in \mathbb{C}^{n}$ and Hermitian covariance matrix $\Sigma \in \mathbb{C}^{n \times n}$ by $CN(\mu, \Sigma)$. By $\mathrm{I}_{n}$, we denote the $n \times n$ identity matrix. Finally, for $x \in \mathbb{C}^{n}$ and $\lambda\in \mathbb{R}^{n}$, we denote by $\|x\|_{\sharp, \lambda}$ the SLOPE norm of $x$, defined as the value $\|x\|_{\sharp, \lambda} = \sum_{i=1}^{n} \lambda_{(i)} |x|_{(i)}$.


\subsection{Signal Model}

Assume the following narrowband signal model,
\begin{equation}\label{our_model}
x(t) = A_{0} \, s(t) + \epsilon(t) + \delta(t).
\end{equation}
The received signal $x(t) \in \mathbb{C}^{m}$ is the output of an $m$ element array in IQ format. The signal of interest $s(t) \in \mathbb{C}^{d}$ contains $d$ sources, where $d < m$. The time-independent matrix $A_{0} \in \mathbb{C}^{m \times d}$ maps the signal $s(t)$ to array measurements based on array geometry, signal frequency, and the time-independent directions of arrival of the signals from each source. We assume throughout that $A_{0}$ is full column rank. The noise term $\epsilon(t)$ is a random $m$ dimensional signal for which $\epsilon(t_{1})$ and $\epsilon(t_{2})$ are independent whenever $t_{1} \neq t_{2}$. The interference term $\delta(t) \in \mathbb{C}^{m}$ is a deterministic signal that is sparse in time, meaning that it is frequently the zero vector.

\subsection{SLOPE and Group SLOPE}

The assumption that $\delta(t)$ in \eqref{our_model} is sparse leads to the following goal: given a set of observations from the sensor array, find the subset of observations for which $\delta(t)$ is nonzero. Since these observations feature an interfering signal of unknown nature, they should not be used to estimate properties of the signal of interest. Therefore estimating the sparsity pattern of $\delta(t)$--the values of $t$ for which it is nonzero--is the primary consideration.

Sparse estimation methods are often computationally burdensome because the sparsity pattern is selected from all subsets of observations, a set which grows exponentially in the number of observations. A popular method to reduce these computational challenges is to estimate sparsity using an optimization problem which features a sparsity-inducing regularizer in addition to a data-fidelity term \cite{wainwright2019high}. The most widely known of these regularized estimation methods is likely the LASSO \cite{tibshirani1996regression}, which adds an $\ell^{1}$ penalty to an objective function (often least-squares) which incentivizes good model fit. Since the LASSO penalty is convex, adding it to a convex objective function results in a convex optimization problem that can be solved quickly for large numbers of observations.

Despite its strengths, using LASSO for sparse estimation does not naturally yield finite sample bounds on the false positive or false negative error rates of the estimator. To overcome these limitations, \cite{bogdan2015slope} introduced the SLOPE penalty, which, for a nonnegative penalty vector $\lambda \in \mathbb{R}^{n}$ having at least one positive entry, is defined as 
\begin{equation}\label{slope_def}
\|x\|_{\sharp, \lambda} = \sum_{i=1}^{n} \lambda_{(i)} |x|_{(i)},
\end{equation}
where $x$ is any vector in $\mathbb{C}^{n}$ and $|x|_{(i)}$ the $i$th order statistic of $(|x_{1}|,\dots, |x_{n}|)$. By taking all entries of $\lambda$ the same, the group SLOPE definition in \eqref{grpSLOPE} reduces to group LASSO. However, these varied $\lambda$ values are key to group SLOPE's advantage over group LASSO as an adaptive penalty, because it results in larger column norms being penalized more heavily.

In the restricted setting of least squares estimation with orthogonal design matrices and additive Gaussian noise, the SLOPE penalty can be shown to control the expected false positive rate of the estimator \cite{bogdan2015slope}. The SLOPE penalty is also convex (see \cite[Proposition 1.2]{bogdan2015slope}), so it retains the attractive computational benefits of the LASSO penalty. Additionally, the authors give a fast algorithm for computing the SLOPE penalty's proximal operator, a computational primitive used in many nonsmooth optimization methods \cite{parikh2014proximal}. The SLOPE penalty's proximal operator is denoted $\mathrm{Prox}_{\|\cdot\|_{\lambda, \sharp}}$ and defined as 
\begin{equation}
\mathrm{Prox}_{\|\cdot\|_{\lambda, \sharp}}(x) = \argmin_{w \in \mathbb{R}^{d}} \; \|x\|_{\lambda, \sharp} + \frac{1}{2} \|x - w \|_{2}^{2}. 
\end{equation} 

In the context of array signal processing, observations are vector valued, so we must use a vector-valued extension to the SLOPE penalty. Let $X \in \mathbb{C}^{m \times n}$ and $\lambda \in \mathbb{R}^{n}$ be a vector of nonnegative values. The group SLOPE norm of $X$, with penalty vector $\lambda$, is defined as the SLOPE norm applied to the vector $\llbracket X \rrbracket$ consisting of the columnwise Euclidean norms of $X$. That is,
\begin{equation}
\left\|\llbracket X \rrbracket \right\|_{\sharp, \lambda} = \sum_{i=1}^{n} \lambda_{(i)} \llbracket X \rrbracket_{(i)}.\label{grpSLOPE}
\end{equation}
The group SLOPE penalty was first introduced in \cite{gossmann2015identification}. In \cite{brzyski2019group} a closely related penalty was shown to satisfy a similar false discovery rate bound as the (non-group) SLOPE penalty, again under the restrictive setting of Gaussian noise and least squares estimation with an orthogonal design matrix.

\subsection{Robust Estimator for the Signal Subspace}

Let $X \in \mathbb{C}^{m \times n}$ denote $n$ observations of $x(t)$ from \eqref{our_model} at times $t_{1},...,t_{n}$, where the observations are stacked as columns. Denote by $S \in \mathbb{C}^{d \times n}$ and $\Delta \in \mathbb{C}^{m \times n}$ decision variables with columns that give the values of the signal of interest $s(t)$ and interference term $\delta(t)$, respectively, for each of the $n$ observations in $X$. For a penalty vector $\lambda$, which will be specified in section \ref{sec:error}, we form estimators $\hat{A}$, $\hat{S}$, and $\hat{\Delta}$ as minimizers of the following group SLOPE regularized least squares problem,
\begin{equation} \label{our_problem}
\min_{A, S, \Delta} \left\|X - A S - \Delta \right\|_{F}^{2} + \left\|\llbracket \Delta \rrbracket \right\|_{\sharp, \lambda}.
\end{equation}
The least squares portion of problem \eqref{our_problem} is a data fidelity term which encourages $\hat{A}$, $\hat{S}$, and $\hat{\Delta}$ to be chosen so that $\hat{A}\, \hat{S} + \hat{\Delta}$ is close to $X$. The group SLOPE term is a regularizer that encourages solution structure based on a priori information on the signal. In this case, the group SLOPE regularizer encourages solutions for which $\hat{\Delta}$ is column sparse, based on the assumption in \eqref{our_model} that $\delta(t)$ is frequently the zero vector.


\subsection{Summary of Main Results}

We solve \eqref{our_problem} using alternating minimization in $\Delta$ and $A$, since $S$ can be eliminated as a decision variable. Alternating minimization is especially promising since its per-iteration complexity is small. We also prove the convergence of alternating minimization to a critical point of problem \eqref{our_problem}.

\textbf{Contribution 1} (Computation).
The sequence of iterates produced by applying alternating minimization in $\Delta$ and $A$ to \eqref{our_problem} is guaranteed to have a limit point, and each limit point is a local minimum of \eqref{our_problem}. Each iteration of alternating minimization (updating both $A$ and $\Delta$) requires $O(n \log(n) + n m^2 + m^3)$ flops.

In addition to considering the numerical error associated with computing the estimator in \eqref{our_problem}, we also consider its statistical error--how well it estimates the unknown quantity of interest. We focus on detecting which observations feature nonzero interference, and measure the estimator's performance by bounding its false discovery rate--the expected proportion of observations which are incorrectly identified as containing interference.

\textbf{Contribution 2} (Statistical Error).
Given a desired false discovery rate $q \in (0,1)$ and a fixed estimate of the signal subspace, there is a choice of the penalty vector $\lambda$ for which the estimator \eqref{our_problem} has false discovery rate bounded by $q$. This $\lambda$ vector depends only on $q$ and the root mean square amplitude of $A_{0} \, s(t) + \epsilon(t)$, the received signal without interference, contained within the estimated noise subspace. Unlike previous false discovery rate results using a group SLOPE penalty \cite{brzyski2019group}, our result holds for non-Gaussian noise and without additional restrictions on the observations.

\section{Computation}\label{sec:computation}

In this section we propose and analyze an iterative algorithm for computing minimizers of \eqref{our_problem}. We prove that sequences produced by the algorithm have limit points, and that these limit points are local minimizers of \eqref{our_problem}.

The fact that these minimizers are local and not global is a result of the nonconvexity in \eqref{our_problem}. However, the local nature of solutions can be advantageous in situations where the signal of interest and interfering signal are not clearly defined a priori. For example, if a practitioner wishes to estimate the signal subspaces of two sparse signals, then by running the algorithm with different initial values the iterates can be made to converge to different critical points which give the signal subspace for each of the signals.

To formulate our algorithm for solving \eqref{our_problem}, we first eliminate the $S$ variable from \eqref{our_problem} by noting that it can be rewritten
\begin{align}
\min_{A, \Delta}\left\{ \min_{S} \left\{\left\|X - A S - \Delta \right\|_{F}^{2}  \right\} + \left\| \llbracket \Delta \rrbracket \right\|_{\sharp, \lambda} \right\}. \label{nested_min}
\end{align}
The inner minimization in $S$ has the closed-form minimizer $\hat{S} = (A^{*} A)^{-1} A^{*} (X - \Delta)$. Substituting this value into \eqref{nested_min}, we have
\begin{equation} \label{red_problem}
\min_{A, \Delta} \left\|(I - P_{A})(X - \Delta) \right\|_{F}^{2} + \left\|\llbracket \Delta \rrbracket \right\|_{\sharp, \lambda},
\end{equation}
where $P_{A} = A (A^{*} A)^{-1} A^{*}$. Since there is a one-to-one correspondence between subspaces and projection matrices \cite{golub2013matrix}, it is natural that problem \eqref{red_problem} uses $P_{A}$ in its objective function instead of $A$. For a given $P_{A}$, the $A$ for which $P_{A} = A (A^{*} A)^{-1} A^{*}$ is nonunique, because any $m \times d$ matrix with left singular vectors that span $\mathrm{col}(P_{A})$ produces the same projection matrix. For this reason, we refer to optimal choices of either $A$ or $P_{A}$ as \emph{solutions} to \eqref{red_problem}, with the understanding that the correspondence between $A$ and $P_{A}$ is well-understood despite being nonunique.

 
We apply alternating minimization to problem \eqref{red_problem}. For fixed $A$, minimization in $\Delta$ is a convex program because the group SLOPE norm is convex. Through a series of reductions, we show that this convex program can be solved by computing a single evaluation of the SLOPE norm's proximal operator, which can be computed in $n \log n$ flops using the algorithm in \cite{bogdan2015slope}. On the other hand, for fixed $\Delta$, minimization in $A$ results in $I - P_{A}$ chosen as the $d$ largest singular vectors of $X - \Delta$. We formalize these results in the following theorems, the proofs of which are postponed to the supplementary material.

\begin{theoremrep}\label{thm:prox}
The problem
\begin{equation}
\min_{\Delta} \left\|(I - P_{A})(X - \Delta) \right\|_{F}^{2} + \left\| \llbracket \Delta \rrbracket \right\|_{\sharp, \lambda}
\end{equation}
has minimizer $\hat{\Delta}$ that can be defined columnwise as
\begin{equation} \label{columnwise_def}
\hat{\Delta}_{i} = \frac{\hat{c}_{i}}{\|(I - P_{A}) X_{i} \|_{2}} (I - P_{A}) X_{i}, \quad \forall i \in [n].
\end{equation}
The vector $\hat{c} \in \mathbb{R}^{n}$ in \eqref{columnwise_def} is defined as
\begin{equation}
\hat{c} = \mathrm{Prox}_{\|\cdot\|_{\sharp, \lambda}}(\llbracket(I - P_A) X\rrbracket),
\end{equation}
the proximal operator of the SLOPE norm applied to the column norms of $(I - P_A) X$.
\end{theoremrep}
\begin{proof}
Consider our objective function when minimizing over $\Delta$ for fixed $A$,
\begin{equation}\label{1st_in_prox_proof}
\min_{\Delta} \sum_{i=1}^{n} \left\|(I - P_{A})(X_{i} - \Delta_{i}) \right\|_{F}^{2} + \sum_{i=1}^{n} \lambda_{(i)} \llbracket \Delta \rrbracket_{(i)}.
\end{equation}
Denote by $\Delta_{(i)}$ the ith largest column of $\Delta$, with the ordering induced by column norms. Applying the Pythagorean theorem, \eqref{1st_in_prox_proof} becomes
\begin{equation}
\min_{\Delta} \sum_{i=1}^{n} \left\|(I - P_{A})(X_{i} - \Delta_{i}) \right\|_{F}^{2} + \sum_{i=1}^{n} \lambda_{(i)} \sqrt{\left\|P_{A} \Delta_{(i)} \right\|_{2}^{2}  + \left\|(I - P_{A})\Delta_{(i)} \right\|_{2}^{2} },
\end{equation}
from which we see that a necessary condition for minimization in $\Delta$ is $P_{A} \Delta_{(i)} = 0$ for all $i \in [n]$. Thus \eqref{1st_in_prox_proof} reduces to
\begin{equation}\label{3rd_in_prox_proof}
\min_{\Delta} \sum_{i=1}^{n} \left\|(I - P_{A})(X_{i} - \Delta_{i}) \right\|_{F}^{2} + \sum_{i=1}^{n} \lambda_{(i)} \llbracket (I - P_{A}) \Delta \rrbracket_{(i)}.
\end{equation}
Introduce the auxiliary vector $c \in \mathbb{R}^{n}$ with $c_{(i)} = \|(I - P_{A}) \Delta_{(i)}\|$, and rewrite \eqref{3rd_in_prox_proof} as the constrained problem
\begin{align}\label{4th_in_prox_proof}
\min_{c, \Delta} & \quad \sum_{i=1}^{n} \left\|(I - P_{A})(X_{i} - \Delta_{i}) \right\|_{F}^{2} + \sum_{i=1}^{n} \lambda_{(i)} \llbracket (I - P_{A}) \Delta \rrbracket_{(i)}\\
\text{s.t.} & \quad c_{(i)} = \|(I - P_{A}) \Delta_{i} \|_{2}, \quad \quad \forall i \in [n].
\end{align}
Expanding the quadratic and neglecting constant terms, we have
\begin{equation} \label{program_in_c}
\min_{c, \Delta} \sum_{i=1}^{n} c_{i}^{2} - 2 \Re\left[\left((I - P_{A}) X_{i}\right)^{*} \left((I - P_{A}) \Delta_{i} \right)\right] + \sum_{i=1}^{n} \lambda_{(i)} c_{(i)}
\end{equation}
For fixed $c$, minimizing in $\Delta$ using the Cauchy-Schwarz inequality gives
\begin{equation}
(I - P_{A}) \Delta_{i} = \frac{c_{i}}{\|(I - P_{A}) X_{i}\|_{2}}(I - P_{A}) X_{i}.
\end{equation}
Since we have previously shown that $P_{A} \Delta_{i} = 0$ is a necessary condition for $\Delta$ to minimize \eqref{1st_in_prox_proof}, we conclude that
\begin{equation} \label{delta_conclusion}
\Delta_{i} = \frac{c_{i}}{\|(I - P_{A}) X_{i}\|_{2}}(I - P_{A}) X_{i}.
\end{equation}
Substituting \eqref{delta_conclusion} into \eqref{program_in_c} yields
\begin{equation}
\min_{c \in \mathbb{R}^{n}} \sum_{i=1}^{n} c_{i}^{2} - 2 c_{i} \left\|(I - P_{A}) X_{i} \right\|_{2} + \sum_{i=1}^{n} \lambda_{(i)} c_{(i)},
\end{equation}
which, completing the square, is equivalent to the following,
\begin{equation}\label{prox_conclusion}
\min_{c \in \mathbb{R}^{n}} \left\|c - \llbracket (I - P_{A}) X \rrbracket \right\|_{2}^{2} + \|c \|_{\sharp, \lambda}.
\end{equation}
The problem \eqref{prox_conclusion} is the proximal operator of $\|\cdot\|_{\sharp, \lambda}$ evaluated at $\llbracket (I - P_{A}) X \rrbracket$. Combined with \eqref{delta_conclusion}, this proves the result.
\end{proof}

\begin{theoremrep}\label{thm:tsvd}
Let $\tilde{U} \tilde{\Sigma} \tilde{V}^{*}$ be a rank-$d$ truncated singular value decomposition of $X - \Delta$. Then $\tilde{U}$ is a minimizer of the problem
\begin{equation}\label{prob_in_A}
\min_{A} \left\|(I - P_{A})(X - \Delta) \right\|_{F}^{2}.
\end{equation}
\end{theoremrep}
\begin{appendixproof}
Denote the full SVD of $X-\Delta$ by $U \Sigma V^{*}$, and the $i$th entry along the diagonal of $\Sigma$ by $\sigma_{i}$. Note that plugging in $\tilde{U}$ to this objective yields 
\begin{align}
\left\|(I - \tilde{U}(\tilde{U}^{*} \tilde{U})^{-1} \tilde{U}^{*} )(X - \Delta) \right\|_{F}^{2}\\
= \left\| (I - \tilde{U} \tilde{U}^{*} ) U \Sigma V^{*} \right\|_{F}^{2}\\
= \left\| (I - \sum_{i=1}^{d} U_{i} U_{i}^{*} ) (\sum_{i=1}^{m} U_{i} \sigma_{i} V_{i}^{*}) \right\|_{F}^{2}\\
= \left\| \sum_{i=1}^{m} U_{i} \sigma_{i} V_{i}^{*} - \sum_{i=1}^{d} U_{i} \sigma_{i} V_{i}^{*} \right\|_{F}^{2}\\
= \left\| \sum_{i=d+1}^{m} U_{i} \sigma_{i} V_{i}^{*}  \right\|_{F}^{2}\\
= \mathrm{Tr}\left[ (\sum_{i=d+1}^{m} U_{i} \sigma_{i} V_{i}^{*} )(\sum_{i=d+1}^{m} V_{i} \sigma_{i} u_{i}^{*}) \right]\\
= \Tr\left[\sum_{i=d+1}^{m} U_{i} \sigma_{i}^{2} U_{i}^{*}\right]\\
= \sum_{i=d+1}^{m} \Tr\left[ \sigma_{i}^{2} U_{i}^{*}U_{i} \right]\\
=\sum_{i=d+1}^{m}\sigma_{i}^{2}.
\end{align}
Therefore,
\begin{equation}
\min_{A \in \mathbb{C}^{m \times d}} \left\|(I - P_{A})(X - \Delta) \right\|_{F}^{2}  \leq \sum_{i=d+1}^{m} \sigma_{i}^{2}.
\end{equation}

We will next show that 
\begin{equation}
\min_{A \in \mathbb{C}^{m \times d}} \left\|(I - P_{A})(X - \Delta) \right\|_{F}^{2} \geq \sum_{i=d+1}^{m} \sigma_{i}^{2},
\end{equation}
which demonstrates that $\tilde{U}$ minimizes \eqref{prob_in_A}.

Recall that an orthogonal projection matrix has the same rank as the dimension of the space it projects onto. Also, for any conformable matrices $A$ and $B$, $\mathrm{rank}(AB) \leq \mathrm{rank}(A)$. Then by the Eckhart-Young Theorem \cite{golub2013matrix},
\begin{align}
& \min_{A \in \mathbb{C}^{m \times d}} \left\|(I - P_{A})(X - \Delta) \right\|_{F}^{2} \\
= & \min_{A \in \mathbb{C}^{m \times d}} \left\|(X - \Delta) - P_{A}(X - \Delta) \right\|_{F}^{2} \\
\geq &\min_{\substack{B \in \mathbb{C}^{m \times n}\\\mathrm{rank}(B) \leq d}} \left\|(X - \Delta) - B \right\|_{F}^{2}\\
= &\sum_{i=d+1}^{m} \sigma_{i}^{2}.
\end{align}
This proves the result.
\end{appendixproof}

Algorithm \ref{altmin} gives our proposed alternating minimization algorithm for solving \eqref{red_problem}, including the simplifications of each of the subproblems provided by Theorems \ref{thm:prox} and \ref{thm:tsvd}. Since computation of the SVD (without truncation) requires
$O(m^2 n + m^3)$ flops \cite{golub2013matrix}, the proximal algorithm of the SLOPE norm from \cite{bogdan2015slope} requires $O(n \log n)$ flops, and the matrix multiplies require $O(m^2 n)$ flops, we conclude that each iteration of Algorithim \ref{altmin} requires $O(n \log n + n m^2 + m^3)$ flops. In practice $n \gg m$, so the $n \log n$ dependence on $n$ is an attractive feature of the algorithm's per-iteration complexity.


\begin{algorithm}
    \caption{Alternating minimization algorithm for \eqref{red_problem}.}
    \label{altmin}
    \begin{algorithmic}[1] 
        \Procedure{AltMin}{X}
	    \State $\hat{\Delta} \gets 0$
            \Repeat
		\State $\hat{A} \gets \textproc{SVD}(X - \hat{\Delta}, d)$ \Comment{Assign $d$ leading left singular vectors to $A$.} 
		\State $P_{\hat{A}} \gets \hat{A} \, \hat{A}^{*}$ 
		\State $\hat{\Delta} \gets \textproc{Prox}_{\|\cdot\|_{\sharp, \lambda}}(\llbracket (I- P_{\hat{A}}) X \rrbracket)$ \Comment{Apply the SLOPE prox to column norms of $(I - P_{\hat{A}}) X$.}
            \Until{convergence}
            \State \textbf{return} $ P_{\hat{A}}, \hat{\Delta}$
        \EndProcedure
    \end{algorithmic}

\vspace{.15cm}

    \begin{algorithmic}[1] 
        \Function{SVD}{$M$, $k$}
	  \State \textbf{return} $\tilde{U}$, where $\tilde{U} \tilde{\Sigma} \tilde{V}^{*}$ is the rank-$k$ truncated singular value decomposition of $M$
        \EndFunction
    \end{algorithmic}

\vspace{.15cm}
    \begin{algorithmic}[1] 
        \State \textbf{function} $\textproc{Prox}_{\|\cdot\|_{\sharp, \lambda}}(w)$
	\State \;\,\, \, \textbf{return} The proximal operator of the SLOPE norm, $\|\cdot\|_{\sharp, \lambda}$, applied to $w$. See \cite{bogdan2015slope}.
	\State \textbf{end function}
    \end{algorithmic}
\end{algorithm}


The initialization of $\hat{\Delta} = 0 $ in Algorithm \ref{altmin} is one of many sensible options. To estimate the interference term of a nondominant signal, one could instead initialize $\hat{A}$ to any matrix for which $\col(\hat{A})$ approximates the signal subspace of interest. The $\hat{\Delta}$ term would be updated to detect interference for that given $\hat{A}$ matrix, and then both $\hat{A}$ and $\hat{\Delta}$ continue to be updated until the algorithm converges. If $\hat{\Delta}^{k}$ and $\hat{A}^{k}$ give the values of $\hat{\Delta}$ and $\hat{A}$ after $k$ iterations of Algorithm \ref{altmin}, then a natural stopping condition is to specify some tolerance $\eta$ and terminate the algorithm when 
\begin{equation} \label{stop_cond}
\|\hat{\Delta}^{k+1} - \hat{\Delta}^{k}\|_{F} < \eta \quad \text{and} \quad \|\hat{A}^{k+1} - \hat{A}^{k}\|_{F} < \eta.
\end{equation}

We next establish convergence theory for Algorithm \ref{altmin}. First, we define our notion of optimality. Let $\mathcal{X}$ be a Hilbert space, and $\mathcal{C} \subseteq \mathcal{X}$. Let $f:\mathcal{X} \to \mathbb{R}$ be a continuously differentiable function. Recall that a point $\hat{x} \in \mathcal{X}$ is said to be a \emph{critical point} of the constrained optimization problem $\min_{x \in \mathcal{C}}\, f(x)$ if $\langle \nabla f(x^{*}), y - x^{*}\rangle \geq 0$, for every $y \in \mathcal{C}$. A \emph{limit point} of a sequence $x^{n} \subseteq X$ is a point $x \in \mathcal{X}$ such that $x^{n}$ has a subsequence that converges to $x$.

The optimization problem \eqref{red_problem}, which defines our estimator, can be reformulated as a constrained problem with continuously differentiable objective by introducing a scalar variable $b$ and writing
\begin{align} \label{constrained_problem}
\min_{A, \Delta, b}  & \quad \left\|(I - P_{A})(X - \Delta) \right\|_{F}^{2} + b\\
\text{s.t.} & \quad \left\|\llbracket \Delta \rrbracket \right\|_{\sharp, \lambda} \leq b. \nonumber
\end{align}
Any optimal solution to this problem will clearly have $b$ attain its lower bound of $\|\Delta\|_{\sharp, \lambda}$, which establishes the equivalence between problems \eqref{red_problem} and \eqref{constrained_problem}. This constrained problem makes stating our next theorem, which gives the convergence properties of Algorithm \ref{altmin}, more convenient. Its proof is deferred to the appendix.

\begin{theoremrep} \label{lim_pts}
The sequence of iterates produced by Algorithm \ref{altmin} has at least one limit point. Each limit point is a critical point of \eqref{constrained_problem}.
\end{theoremrep}
\begin{proof}
The proof will proceed in the following steps. First, we note that
\begin{equation} \label{initialization_bound}
(P_{A}^{k}, \Delta^{k})_{k=1}^{\infty} \subseteq \left\{(P_{A}, \Delta):\left\|(I - P_{A})(X - \Delta) \right\|_{F}^{2} +  \|\llbracket \Delta \rrbracket \|_{\sharp, \lambda} \leq \left\|(I - P_{A^{0}})(X - \Delta^{0}) \right\|_{F}^{2} + \| \llbracket \Delta^{0} \rrbracket \|_{\sharp, \lambda} \right\}.
\end{equation}
for any choice of initialization $(P_{A}^{0}, \Delta^{0})$. Inequality \eqref{initialization_bound} follows from the fact that the objective function in \eqref{constrained_problem} is nonincreasing after each update of $\Delta$ and $P_{A}$. Denote by $\mathrm{Gr}(d, m) \subseteq \mathbb{C}^{m \times m}$ the set of $m \times m$ orthogonal projection matrices which project onto a subspace of dimension of $d$ (this is the set of Hermitian, idempotent matrices $M$ for which $\dim(\col(M))=d$). We show that for any $\alpha \in \mathbb{R}$,
\begin{equation}
\left\{(P_{A}, \Delta): \left\|(I - P_{A})(X - \Delta) \right\|_{F}^{2} +  \|\llbracket \Delta \rrbracket \|_{\sharp, \lambda} \leq \alpha,  \quad  P_{A}  \in \mathrm{Gr}(d,m)\right\}
\end{equation}
is compact. By the Heine Borel theorem, we can conclude that $(P_{A}^{k}, \Delta^{k})_{k=1}^{\infty}$ has a limit point. Finally, to show that each limit point is a critical point, we reformulate the alternating minimization subproblem in $P_{A}$ over a convex constraint set, so that \cite[Corollary 2]{grippo2000convergence} gives that any limit point of alternating minimization is a critical point.

Fix $\alpha \in \mathbb{R}$. We want to show that 
\begin{equation}\label{to_show_compact}
\left\{(P_{A}, \Delta): \left\|(I - P_{A})(X - \Delta) \right\|_{F}^{2} +  \| \llbracket \Delta \rrbracket \|_{\sharp, \lambda} \leq \alpha,  \quad  P_{A}  \in \mathrm{Gr}(d,m)\right\}
\end{equation}
is compact. Since the empty set is compact, it suffices to only consider $\alpha$ such that the set in \eqref{to_show_compact} is nonempty. Since this set is finite-dimensional, it suffices to show that it is closed and bounded. The set in \eqref{to_show_compact} is compact because it can be written as the intersection of sets,
\begin{equation}\label{set_intersection}
\left\{(P_{A}, \Delta): \left\|(I - P_{A})(X - \Delta) \right\|_{F}^{2} +  \| \llbracket \Delta \rrbracket \|_{\sharp, \lambda} \leq \alpha \right\} \bigcap \left\{(P_{A}, \Delta): P_{A}  \in \mathrm{Gr}(d,m)\right\}.
\end{equation}
We will show that the leftmost of these sets is the intersection are compact and that the right set is closed, so that their intersection is compact. The left set is closed because it is the lower level set of a continuous (and hence lower semicontinuous) function. Assume, for contradiction, that it is not bounded. Then there exists a sequence $(P_{A}^{k}, \Delta^{k})_{k=1}^{\infty}$ such that $\|P_{A}^{k}\|_{F} + \|\Delta^k\|_{F} \to \infty$ with 
\begin{equation}
\left\|(I - P_{A}^{k})(X - \Delta^{k}) \right\|_{F}^{2} +  \|\llbracket \Delta^{k} \rrbracket \|_{\sharp, \lambda} \leq \alpha \quad \quad \forall k. \label{the_alpha_bound}
\end{equation}
If $\|\Delta^{k}\|_{F} \to \infty$, then the fact there exists one entry of the $\lambda$ vector which is greater than $0$ gives
\begin{equation}
\left\|(I - P_{A}^{k})(X - \Delta^{k}) \right\|_{F}^{2} +  \|\llbracket \Delta^{k} \rrbracket \|_{\sharp, \lambda} \to \infty,
\end{equation}
which contradicts \ref{the_alpha_bound} and shows that $\|\Delta^{k}\|_{F}$ must be bounded. On the other hand, the value of $\|P_{A}^{k}\|_{F}$ is fixed at $\sqrt{d}$ for all $k$. Therefore $\|P_{A}^{k}\| + \|\Delta^k\| \not\to \infty$. Since no unbounded sequence exists we conclude that the left set in the intersection \eqref{set_intersection} is bounded. Having already shown that it is closed, we conclude that it is compact. The right set from the intersection \eqref{set_intersection} is closed because it can be written as the intersection of the preimage of $0$ for the three continuous maps: $M \to M^{2} - M$, $M \to M^{*} - M$ and $M \to \Tr(M) - d$.

The final point to show is that \eqref{constrained_problem} can be formulated as a continuously differentiable objective over a cartesian product of closed, nonempty, convex constraint sets. This allows us to apply \cite[Corollary 2]{grippo2000convergence} to conclude that any limit point of two-block alternating minimization is a critical point. By Lemma \ref{convex_rewrite}, variable $(\hat{\Delta}, \hat{b}, \hat{P_{A}})$ minimizes \eqref{constrained_problem} (with decision variable $P_{A}$ instead of $A$) if and only if $(\hat{\Delta}, \hat{b}, \hat{P_{A}})$ minimizes
\begin{align}
\min_{P_{A}, \Delta, b}  & \quad \left\|(I - P_{A})(X - \Delta) \right\|_{F}^{2} + b\\
\text{s.t.} & \quad \left\|\llbracket \Delta \rrbracket \right\|_{\sharp, \lambda} \leq b \nonumber\\
& \hspace{-1.75cm}P_{A} = P_{A}^{*},\; 0 \preceq P_{A} \preceq \mathrm{I}_{m}, \; \Tr(P_{A}) = d, \nonumber
\end{align}
where $\preceq$ denotes the Loewner ordering. The set
\begin{equation}
\left\{P_{A} \in \mathbb{C}^{m \times m}: P_{A} = P_{A}^{*},\; 0 \preceq P_{A} \preceq \mathrm{I}_{m}, \; \Tr(P_{A}) = d \right\}
\end{equation}
is convex, nonempty, and closed. Because the SLOPE norm is continuous and convex,
\begin{equation}
\left\{(\Delta, b) \in \mathbb{C}^{m \times n} \times \mathbb{C}^{n}: \left\|\llbracket \Delta \rrbracket \right\|_{\sharp,\lambda} \leq b \right\}
\end{equation}
is convex, nonempty, and closed. Therefore \cite[Corollary 2]{grippo2000convergence} applies and we conclude that any limit point of Algorithm \ref{altmin} is a critical point.
\end{proof}

\begin{toappendix}
\begin{lemma}\label{convex_rewrite}
The following problems have the same minimizers

\begin{minipage}{.5\textwidth}
\begin{align}
\min_{P_{A}, \Delta, b}  & \quad \left\|(I - P_{A})(X - \Delta) \right\|_{F}^{2} + b \nonumber\\
\text{s.t.} & \quad \left\|\llbracket \Delta \rrbracket \right\|_{\sharp, \lambda} \leq b \nonumber\\
& \; \quad P_{A} \in \mathrm{Gr}(d,m). \nonumber
\end{align}
\end{minipage}%
\begin{minipage}{.5\textwidth}
\begin{align}
\min_{P_{A}, \Delta, b}  & \quad \left\|(I - P_{A})(X - \Delta) \right\|_{F}^{2} + b \nonumber\\
\text{s.t.} & \quad \left\|\llbracket \Delta \rrbracket \right\|_{\sharp, \lambda} \leq b \nonumber\\
& \hspace{-1.75cm} 0 \preceq P_{A} \preceq I, \; \Tr(P_{A}) = d, \; P_{A} = P_{A}^{*}. \nonumber
\end{align}
\end{minipage}
\end{lemma}
\begin{proof}
Recall that the set of $m \times m$ orthogonal projection matrices which project onto a $d$ dimensional subspace have $d$ eigenvalues which are $1$ with the remaining $m-d$ eigenvalues $0$. Therefore the following sequence of problems is a relaxation of the original,
\begin{align}
        & \min_{\substack{P_{A} \in \mathbb{C}^{m \times m}\\ P_{A}^{2} =  P_{A},\, P_{A}^{*} =  P_{A}\\ \dim(\col(P_{A})) = d}}
        \quad \min_{\substack{c \in \mathbb{R} \\\Delta \in \mathbb{C}^{m \times n}\\                    \|\llbracket \Delta \rrbracket \|_{\sharp, \lambda} \leq b}} \quad \|(I - P_{A}) (\Delta - X)\|_{F}^{2} + b\\
        = & \min_{\substack{P_{A} \in \mathbb{C}^{m \times m}\\ P_{A}^{2} =  P_{A},\, P_{A}^{*} =  P_{A}\\ \dim(\col(P_{A})) = d}}
        \quad \min_{\substack{c \in \mathbb{R} \\\Delta \in \mathbb{C}^{m \times n}\\                    \|\Delta\|_{\sharp, \lambda} \leq b}} \quad \|(\Delta - X)\|_{F}^{2} -  \Tr\left[P_{A} (\Delta -         X)(\Delta - X)^{*}\right] + b \label{b4_relax}\\
        \geq &
        \min_{\substack{P_{A} \in \mathbb{C}^{m \times m}\\ 0 \preceq P_{A} \preceq  \mathrm{I}_{m} \\
                        \Tr\left(P_{A}\right) = d}}
        \quad \min_{\substack{c \in \mathbb{R} \\\Delta \in \mathbb{C}^{m \times n}\\                    \|\Delta\|_{\sharp, \lambda} \leq b}} \quad \|(\Delta - X)\|_{F}^{2} -  \Tr\left[P_{A} (\Delta -         X)(\Delta - X)^{*}\right] + b. \label{relax}
\end{align}
The equality in \ref{b4_relax} is directly from the definition of the Frobenius inner product, and the inequality in \ref{relax} follows because we have induced a relaxation on the set of projection matrices. Additionally, we both problems \eqref{b4_relax} and \eqref{relax} attain the same lower bound given by the equality condition of the von Neumann trace inequality (see e.g. \cite{mirsky1975trace}). This allows us to conclude that the inequality in \eqref{relax} is actually an equality, with both problems having the same minimizers.
\end{proof}
\end{toappendix}


We can directly interpret the critical point condition of the constrained problem from \eqref{constrained_problem} in the context of the unconstrained problem from \eqref{red_problem}.  Denote the objective function in \eqref{constrained_problem} by $f(A, \Delta) + b$, and denote by $\nabla_{A}$ and $\nabla_{\Delta}$ the gradients with respect to $A$ and $\Delta$, respectively. Then the critical point condition is 
\begin{equation} \label{cp_cond}
\langle \nabla_{A} f(\hat{A}, \hat{\Delta}), A - \hat{A} \rangle + \langle \nabla_{\Delta} f(\hat{A}, \hat{\Delta}),  \Delta - \hat{\Delta} \rangle + b - \hat{b} \geq 0,
\end{equation}
for all $(b, \Delta)$ with $b \geq \left\|\llbracket \Delta \rrbracket\right\|_{\sharp, \lambda}$. A necessary condition for optimality is $\hat{b} = \|\llbracket \hat{\Delta} \rrbracket \|_{\sharp, \lambda}$, so \eqref{cp_cond} implies
\begin{align} & \langle \nabla_{A} f(\hat{A}, \hat{\Delta}), A - \hat{A} \rangle + \langle \nabla_{\Delta} f(\hat{A}, \hat{\Delta}),  \Delta - \hat{\Delta} \rangle \label{Tseries}\\
& \quad\quad + \|\llbracket \Delta \rrbracket\|_{\sharp, \lambda} - \|\llbracket \hat{\Delta} \rrbracket \|_{\sharp, \lambda} \geq 0. \label{constraint}
\end{align}
Line \eqref{Tseries} is a first-order approximation of $f(A, \Delta) - f(\hat{A}, \hat{\Delta})$. 
Using a Taylor series expansion, the inequality in \eqref{Tseries}-\eqref{constraint} gives that 
\begin{align}
f(A, \Delta) - f(\hat{A}, \hat{\Delta}) + \|\llbracket \Delta \rrbracket \|_{\sharp, \lambda} - \|\llbracket \hat{\Delta} \rrbracket \|_{\sharp, \lambda} \\
\geq o(\|A - \hat{A}\|_{F} + \|\Delta - \hat{\Delta}\|_{F}).
\end{align}
So, as $\|A - \hat{A}\|_{F} + \|\Delta - \hat{\Delta}\|_{F} \to 0$,
\begin{equation}\label{no_local_descent}
\liminf f(A,\Delta) + \|\llbracket \Delta \rrbracket \|_{\sharp, \lambda} \geq f(\hat{A}, \hat{\Delta}) + \| \llbracket \hat{\Delta} \rrbracket \|_{\sharp, \lambda}.
\end{equation}
Inequality \eqref{no_local_descent} demonstrates that any critical point of \eqref{constrained_problem} has no local descent directions for problem \eqref{red_problem}, so it is a local minimizer. Therefore, by Theorem \ref{lim_pts}, each limit point produced by Algorithm \ref{altmin} is a local minimizer of \eqref{red_problem}.


\section{Statistical Performance}\label{sec:error}

In this section, we examine the statistical error of estimators derived from minimizing \eqref{red_problem}. Since our primary goal is to identify observations which feature interference, we focus on the support of $\hat{\Delta}$--the set of columns which are not estimated to be the zero vector. The problem of correctly identifying whether each observation in a sample contains nonzero interference is a sequence of hypothesis tests.  In the context of detecting interfered observations, we take the null hypothesis to be the assertion that a certain sample does not contained interference, and take as alternative hypothesis the assertion that the interference term is nonzero for that sample. Since larger penalty values encourage sparsity, selecting the entries of $\lambda$ larger results in an estimate $\hat{\Delta}$ with more sparsity, so that fewer observations are estimated to contain interference. On the other hand, choosing the $\lambda$ vector with smaller entries results in an estimate $\hat{\Delta}$ with more nonzero entries, meaning that more observations are identified as containing interference. An ideal $\lambda$ vector is one which achieves an appropriate balance between false positive errors, where observations are incorrectly identified as containing interference, and false negative errors, where observations are incorrectly identified as containing no interference.

A common way to measure error across a sequence of hypothesis tests is the \emph{false discovery rate}, abbreviated FDR \cite{efron2021computer}. The FDR is defined as the expected false positive rate--the number of incorrectly discarded observations divided by the total number of discarded observations. If $\tilde{\Delta}$ is any estimate of a true interference term $\Delta_{0}$, then the \emph{false discovery rate} of $\tilde{\Delta}$ is
\begin{equation}\label{fdr}
\mathrm{FDR} = \mathbb{E}\left[\frac{\left|\left\{i \in [n] : \llbracket \Delta_{0} \rrbracket_{i} = 0 \text{ and }  \llbracket \tilde{\Delta} \rrbracket_{i} \neq 0\right\}\right|}{\max\left\{ \left|\left\{i \in [n]: \llbracket \tilde{\Delta} \rrbracket_{i} \neq 0\right\}\right|, 1\right\}}\right].
\end{equation}

Our next result gives a prescription for the $\lambda$ vector in \eqref{red_problem} based on a desired false discovery rate. This yields an estimator which has a false discovery rate that is bounded above by the desired quantity. 
\begin{theoremrep}\label{fdr_theorem}
Consider $n$ observations from the statistical model \eqref{our_model}, of which $n_{0}$ have $\delta(t) = 0$. Let $\hat{A}$ be a nonrandom estimate of $A_{0}$ and $q \in (0,1)$ the desired false discovery rate. Take the parameter vector $\lambda$ in the group SLOPE penalty such that
\begin{equation}\label{lambda_recipe}
\lambda_{(n - r + 1)} = \max_{i \in [n]} \Big\{ F^{-1}_{\|((I - P_{\hat{A}})(A_{0} s(t_{i}) + \epsilon(t_{i}))\|}\Big(1-\frac{q \, r}{n}\Big)\Big\},
\end{equation}
where $F^{-1}_{\|((I - P_{\hat{A}})(A_{0} s(t_{i}) + \epsilon(t_{i}))\|}$ denotes the quantile function of the random variable $\|((I - P_{\hat{A}})(A_{0} s(t_{i}) + \epsilon(t_{i}))\|$. Then the false discard rate for testing $\Delta \neq 0$ is bounded above via 
\begin{equation} \label{fdr_bound}
\mathrm{FDR} \leq q \frac{n_{0}}{n}.
\end{equation}
\end{theoremrep}
\begin{proof}
From Theorem \ref{thm:prox}, the FDR in \eqref{fdr} can be rewritten in terms of the auxiliary variable 
$c$ as 
\begin{equation}\label{fdr_in_c}
\mathrm{FDR} = \mathbb{E}\left[\frac{\left|\left\{i \in [n]: c_{i} = 0, \hat{c}_{i} \neq 0\right\}\right|}{\max\left\{ \left|\left\{i \in [n]: \hat{c}_{i} \neq 0\right\}\right|, 1\right\}}\right],
\end{equation}
where $c_{i} = \llbracket \Delta_{0} \rrbracket_{i}$. We first write \eqref{fdr_in_c} as a sum of probabilities as follows, where we denote by $\mathcal{I}$ the set $\{i \in [n]: c_{i} \neq 0\}$,
\begin{align}
&\quad \quad \mathbb{E}\left[\frac{\left|\left\{i \in [n]: c_{i} = 0, \hat{c}_{i} \neq 0\right\}\right|}{\max\left\{ \left|\left\{i \in [n]: \hat{c}_{i} \neq 0\right\}\right|, 1\right\}}\right]\\
&=\sum_{r=1}^{n} \mathbb{E}\left[ \frac{\left|\left\{i: c_{i} = 0, \hat{c}_{i} \neq                      0\right\}\right|}{r}\, \mathbbm{1}_{\max\left\{ \left|\left\{i \in [n]: \hat{c}_{i} \neq 0\right\}\right|, 1\right\} = r}\right]\\
&=\sum_{r=1}^{n} \frac{1}{r} \mathbb{E}\left[ \sum_{i \in [n]\setminus \mathcal{I}}                      \mathbbm{1}_{\left\{\hat{c}_{i} \neq 0\right\}} \, \mathbbm{1}_{\left\{\max\left\{ \left|\left\{j \in [n]: \hat{c}_{j} \neq 0\right\}\right|, 1\right\} = r\right\}}\right] \label{expand_in_i}\\
&=\sum_{r=1}^{n} \frac{1}{r}  \sum_{i \in [n]\setminus \mathcal{I}}                                      \mathbb{E}\left[\mathbbm{1}_{\left\{\hat{c}_{i} \neq 0\right\}} \, \mathbbm{1}_{\left\{\max\left\{       \left|\left\{j \in [n]: \hat{c}_{j} \neq 0\right\}\right|, 1\right\} = r\right\}}\right]\\
&=\sum_{r=1}^{n} \frac{1}{r}  \sum_{i \in [n]\setminus \mathcal{I}} \mathbb{P}\Big(\hat{c}_{i} \neq 0   \text{ and } \max\left\{ \left|\left\{j \in [n]: \hat{c}_{j} \neq 0\right\}\right|, 1\right\} = r\Big). \label{from_expect_to_prob}
\end{align}
Note that in \eqref{expand_in_i} and what follows, the summation is over all $i$ for which $c_{i} =0$. For fixed $i$, we examine the expression
\begin{equation} \label{fixed_i}
\mathbb{P}\Big(\hat{c}_{i} \neq 0   \text{ and } \max\left\{ \left|\left\{j \in [n]: \hat{c}_{j} \neq 0\right\}\right|, 1\right\} = r\Big).
\end{equation}
This is probability of the event where the $i$th sample is estimated to contain interference and $r$ observations in total are estimated to contain interference. By Lemma B.1 in \cite{bogdan2015slope},\begin{align}
& \mathbb{P}\Big(\hat{c}_{i} \neq 0   \text{ and } \max\left\{ \left|\left\{j \in [n]: \hat{c}_{j} \neq 0\right\}\right|, 1\right\} = r\Big)\\
\leq \; & \mathbb{P}\Big(\llbracket(I - P_{\hat{A}}) X \rrbracket_{i} >  \lambda_{(n-r+1)} \text{ and } \max\left\{ \left|\left\{j \in [n]: \hat{c}_{j} \neq 0\right\}\right|, 1\right\} = r \Big).
\label{from_r_to_rm1}
\end{align}
From Lemma B.2 in \cite{bogdan2015slope}, the expression in \eqref{from_r_to_rm1} is bounded above by
\begin{equation}
\leq \; \mathbb{P}\Big(\llbracket(I - P_{\hat{A}}) X \rrbracket_{i} >  \lambda_{(n-r+1)} \text{ and } \max\left\{ \left|\left\{j \in [n]: \hat{c}_{j} \neq 0 \text{ and } j \neq i \right\}\right|, 1\right\} = r-1 \Big). \label{after_rm1}
\end{equation}
Note that in \eqref{fixed_i} we have restricted ourselves to $i$ for which $c_{i} = 0$, so that $X_{i}$ contains no interference. In this setting, $X_{i} = A_{0}\,s(t_{i}) + \epsilon(t_{i})$, so \eqref{after_rm1} becomes
\begin{equation}
\end{equation}

Since $\hat{A}$ is fixed, and $\epsilon$ is independent across observations, 
\begin{align}
& \mathbb{P}\Big(\|(I - P_{\hat{A}})(A_{0} s(t_{i}) + \epsilon(t_{i})) \|_{2} >  \lambda_{(n-r+1)} \text{ and } \max\left\{ \left|\left\{j \in [n]: \hat{c}_{j} \neq 0 \text{ and } j \neq i \right\}\right|, 1\right\} = r-1 \Big) \\
= \; & \mathbb{P}\Big(\|(I - P_{\hat{A}})(A_{0} s(t_{i}) + \epsilon(t_{i})) \|_{2} >  \lambda_{(n-r+1)}\big) \, \mathbb{P}\Big(\max\left\{ \left|\left\{j \in [n]: \hat{c}_{j} \neq 0 \text{ and } j \neq i \right\}\right|, 1\right\} = r-1 \Big) \\
\leq \; & \frac{q r}{n} \, \mathbb{P}\Big(\max\left\{ \left|\left\{j \in [n]: \hat{c}_{j} \neq 0 \text{ and } j \neq i \right\}\right|, 1\right\} = r-1 \Big)
\end{align}
Plugging this bound into \eqref{from_expect_to_prob}, we have
\begin{align}
&\quad \quad \mathbb{E}\left[\frac{\left|\left\{i \in [n]: c_{i} = 0, \hat{c}_{i} \neq 0\right\}\right|}{\max\left\{ \left|\left\{i \in [n]: \hat{c}_{i} \neq 0\right\}\right|, 1\right\}}\right]\\
\leq\; & \sum_{r=1}^{n} \frac{1}{r} \sum_{i \in [n] \setminus \mathcal{I}} \frac{q r}{n} \, \mathbb{P}\Big(\max\left\{ \left|\left\{j \in [n]: \hat{c}_{j} \neq 0 \text{ and } j \neq i \right\}\right|, 1\right\} = r-1 \Big)\\
=\; & \frac{q}{n} \sum_{i \in [n] \setminus \mathcal{I}} \sum_{r=1}^{n} \mathbb{P}\Big(\max\left\{ \left|\left\{j \in [n]: \hat{c}_{j} \neq 0 \text{ and } j \neq i \right\}\right|, 1\right\} = r-1 \Big)\\
=\; & \frac{q \, n_{0}}{n}.
\end{align}
\end{proof}

Theorem \ref{fdr_theorem} gives that, for a fixed estimate of $A_{0}$, the $\lambda$ vector should be chosen based on the desired FDR bound and the root mean square of the signal within the estimated noise subspace \emph{assuming no interference}. In particular, the $\lambda$ vector should be chosen via quantiles of the norm of the signal contained within the estimated noise subspace to obtain the desired FDR guarantee. In the case that $\epsilon(t_{i})$ is a mean zero and isotropic complex Gaussian random variable, these are the quantiles of a scaled noncentral $\chi$ distribution. Considering the optimistic setting where $\hat{A}$ perfectly estimates $A_{0}$, we have $(I - P_{\hat{A}})\, A_{0} s(t_{i})  = 0$ so the $A_{0} s(t_{i})$ term in \eqref{lambda_recipe} vanishes and the quantiles can be chosen from a scaled $\chi$ distribution with degree of freedom equal to $2(m - d)$. The scaling of the $\chi$ distribution is then only a function of the scalar-valued scaling of the identity matrix in the covariance of $\epsilon(t_{i})$. We summarize this important special case in the following corollary.

\begin{corollaryrep}\label{chi_lambdas}
In the context of Theorem \ref{fdr_theorem}, assume $\hat{A} = A_{0}$ and $\epsilon(t_{i}) \sim CN(0, \sigma^{2} \mathrm{I}_{m})$ for all $i \in [n]$. Take the parameter vector $\lambda$ in the group slope penalty such that
\begin{equation} \label{cor_lambda}
\lambda_{(n-r+1)} = \frac{\sigma \sqrt{2}}{2} \,F^{-1}_{\chi_{k}} \left(1 - \frac{q \, r}{n}\right),
\end{equation}
where $F^{-1}_{\chi_{k}}$ denotes the quantile function of a $\chi$ distributed random variable with $k=2(m-d)$ degrees of freedom. Then the bound on the false discard rate in \eqref{fdr_bound} holds.
\end{corollaryrep}
\begin{proof}
To ease notation, denote by $\epsilon$ be a $CN(0, \sigma^{2} I_{m})$ random vector. In the case that $\hat{A} = A_{0}$,
\begin{equation}
\left\|(I - P_{\hat{A}})(A_{0} s(t_{i}) + \epsilon(t_{i}))\right\|_{2} = \left\|(I - P_{\hat{A}}) \epsilon \right\|_{2}.
\end{equation}
We show that $\frac{2}{\sigma^{2}} \left\|(I - P_{\hat{A}}) \epsilon \right\|^{2}$ is $\chi^{2}_{2(m-d)}$ distributed, from which the result follows by applying Theorem \ref{fdr_theorem}. Let $z \sim N(0, \mathrm{I}_{m-d})$ be a standard complex normal random vector of dimension $m-d$, and let $U \begin{pmatrix} \mathrm{I}_{m-d} & 0\\ 0 & 0 \end{pmatrix} U^{*}$ be an eigendecomposition of $I - P_{\hat{A}}$ such that the last $d$ orthonormal columns $u_{m-d}, \dots, u_{m}$ span $\mathrm{col}(A)$. Since the matrix $U$ is orthogonal, $U^{*} z$ is $CN(0,\mathrm{I}_{m})$, with the result that each $u_{i}^{*} z \sim CN(0, 1)$ and $\{\Re(u_{1}^{*} z), \Im(u_{1}^{*} z), \dots, \Re(u_{m}^{*} z), \Im(u_{m}^{*} z)\}$ is a collection of independent $N(0, \frac{\sqrt{2}}{2})$ standard normal random variables.
We then have
\begin{align}
\left\|(I - P_{\hat{A}}) \epsilon\right\|^{2}_{2} &= \epsilon^{*} U \begin{pmatrix} \mathrm{I}_{m-d} & 0\\ 0 & 0 \end{pmatrix} U^{*} \epsilon\\
&= \sigma^{2} \sum_{i=1}^{m-d} \|u_{i}^{*} z\|_{2}^2 \\
&= \sigma^{2} \sum_{i=1}^{m-d} \Re(u_{i}^{*} z)^{2} + \Im(u_{i}^{*} z)^{2}\\
&= \frac{\sigma^{2}}{2} \sum_{i=1}^{m-d} \left(\Re\left\{u_{i}^{*} z\right\}\sqrt{2} \right)^{2} + \left(\Im\left\{u_{i}^{*} z\right\}\sqrt{2} \right)^{2}
\end{align}

Therefore $\frac{2}{\sigma^{2}} \left\|(I - P_{\hat{A}}) \epsilon\right\|_{2}^{2}$ is the sum of $2(m-d)$ independent standard normal random variables, so it is distributed $\chi^{2}_{2(m-d)}$. Taking the square root of $\frac{2}{\sigma^{2}} \left\|(I - P_{\hat{A}}) \epsilon\right\|_{2}^{2}$ proves the result.
\end{proof}



\begin{remark}
The $\sqrt{2}/2$ scaling in \eqref{cor_lambda} accounts for the fact that the real and imaginary components of a standard complex Gaussian are independent Gaussian variables with zero mean and variance $1/2$. See \cite{goodman1963statistical} for details. If the real and imaginary components of $\epsilon$ are independent $N(0, \sigma^{2} \mathrm{I}_{m})$ random vectors, then \eqref{cor_lambda} becomes
\begin{equation}
\lambda_{(n-r+1)} = \sigma \,F^{-1}_{\chi_{k}} \left(1 - \frac{q \, r}{n}\right).
\end{equation}

\end{remark}

In the experiments that follow, we will choose the $\lambda$ vector according to Corollary \ref{chi_lambdas}. From the assumptions required for this result, we see that it may fail in cases where our estimate $\hat{A}$ is far from $A_{0}$. In those cases, the algorithm has converged to a stationary point far from $A_{0}$. As we describe in the next section, this failure case is unlikely and only occurs when a high proportion of the observations feature interference.

\section{Experiments}\label{sec:experiments}

In this section, we perform a number of experiments designed to demonstrate the statistical and computational efficiency of the estimator proposed in \eqref{red_problem} and computed with Algorithm \ref{altmin}. We choose the penalty vector $\lambda$ vector according to \eqref{cor_lambda} and the stopping condition from \eqref{stop_cond} with $\eta = 10^{-6}$. For the sake of scientific reproducibility, a program to reproduce all results in this section can be found on the first author's website.

Throughout this section our experiments will take the following form; simulate 
\begin{equation}
x(t) = A_{0} \, s(t) + \epsilon(t) + \delta(t)
\end{equation}
under the assumption of one target signal of interest and a $50$ channel array. Therefore $s(t) \in \mathbb{C}$, $x(t) \in \mathbb{C}^{50}$, and $A_{0} \in \mathbb{C}^{50 \times 1}$. We set additional sensor and signal parameters according to typical values one might encounter in underwater acoustics \cite{rossing2007springer}, though our methodology is not limited to this domain. We take the signal of interest to be a $300$ Hz sine wave with unit amplitude, and we assume that the sensor array has a sampling rate of 10 kHz. The noise term will be independent and identically distributed isotropic complex Gaussian, $\epsilon(t) \sim CN(0, \sigma_{\epsilon}^{2} \mathrm{I}_{50})$, for each $t$. Throughout the following subsections, we take $\sigma_{\epsilon} = \sqrt{2}/2$, so that the signal to noise ratio \emph{without interference} is $10 \log(2)$ dB. We will change the intensity of the interference term $\delta(t)$, its sparsity level, and its nature (by considering both directed and undirected interference) throughout our experiments.

We construct $A_{0}$ from the geometry of a uniform linear array with quarter-wavelength spacing. In order to have an easily interpretable metric for our ability to estimate a signal subspace, we will focus on estimating the direction of arrival of the signal of interest. Throughout, we take the true direction of arrival for the $300$ Hz signal of interest to be $\pi/4$ radians from the axis defining the sensor array. After using the estimator in \eqref{red_problem} to estimate the observations with nonzero interference, we produce a direction of arrival estimate by using the conventional signal subspace estimator \eqref{opt_problem1} with those observations discarded.




\subsection{Random Interference}\label{sec:randint}

In this section we choose parameters $p \in [0,1]$ and $\sigma_{\delta} > 0$, and then for each $t$ take $\delta(t) = r(t) \xi(t)$, where $r(t) \sim \mathrm{Bernoulli}(p)$ and $\xi(t) \sim CN(0, \sigma_{\delta}^{2} \mathrm{I}_{50})$. Phrased differently, interference is nonzero with probability $p$, and nonzero interference takes the form of a mean zero and isotropic complex Gaussian with variance parameter $\sigma_{\delta}^{2}$. We take $q$, the desired False Discovery Rate, to be $0.1$, and generate the penalty vector $\lambda$  according to Corollary \ref{chi_lambdas}. We assume $10$ seconds worth of data from the sensor array (corresponding to $100,000$ vector-valued observations). 

The estimator in \eqref{red_problem} performs very well in this setting for a wide range of parameters  $\sigma_{\delta}$ and $p$. For an initial example, taking $\sigma_{\delta} = \sqrt{2}$ and $p = .33$ results in an estimator which has confusion matrix given in Table \ref{tab:randint-conmat}. It is remarkable that the estimator produced no false negatives. The false positive rate for this example is $\approx 0.066$, well below our desired false discard rate of $q = 0.1$.

\begin{table}[h]
\centering
\begin{tabular}{ll|rr}
\toprule
{} & {} &  \multicolumn{2}{c}{Estimation} \\
{} & {} &  No Interference &  Interference \\
\midrule
\multirow{2}{*}{Truth} & No Interference &                    64726 &                  2326 \\
{} & Interference    &                        0 &                 32948 \\
\bottomrule
\end{tabular}
\caption{$\sigma_{\epsilon} = \sqrt{2}/2$, $\sigma_{\delta} = \sqrt{2}$, $p = .33$. In this example, the true direction of arrival was $\pi/4$ and the estimate produced after discarding observations estimated to contain interference was $0.2497 \pi$.}
\label{tab:randint-conmat}
\end{table}

For $\sigma_{\epsilon} = \sqrt{2}/2$ and $\sigma_{\delta} = \sqrt{2}$, Figure \ref{fig:fnr-fpr-in-p} gives false positive and false negative rates as a function of $p$. Figure \ref{fig:fnr-fpr-in-sigma} gives the false positive and false negative rate as a function of $p$. In both figures the false negative rate is extremely small, while the false positive rate is well below the desired rate of $q = 0.1$. The case for varying $p$ is especially interesting because one might expect that more interference (larger $p$) would lead to more errors. Instead, the false positive rate decreases linearly as the proportion of interfered observations increases, which is consistent with role of $n_{0}$ in the FDR bound \eqref{fdr_bound}. Large values of $\sigma_{\delta}$, which result in a small signal to interference ratio, have no visible effect on either the false positive or false negative rates.

\begin{figure}
\centering
\begin{subfigure}[a]{.5\textwidth}
\centering
\includegraphics[width=3in]{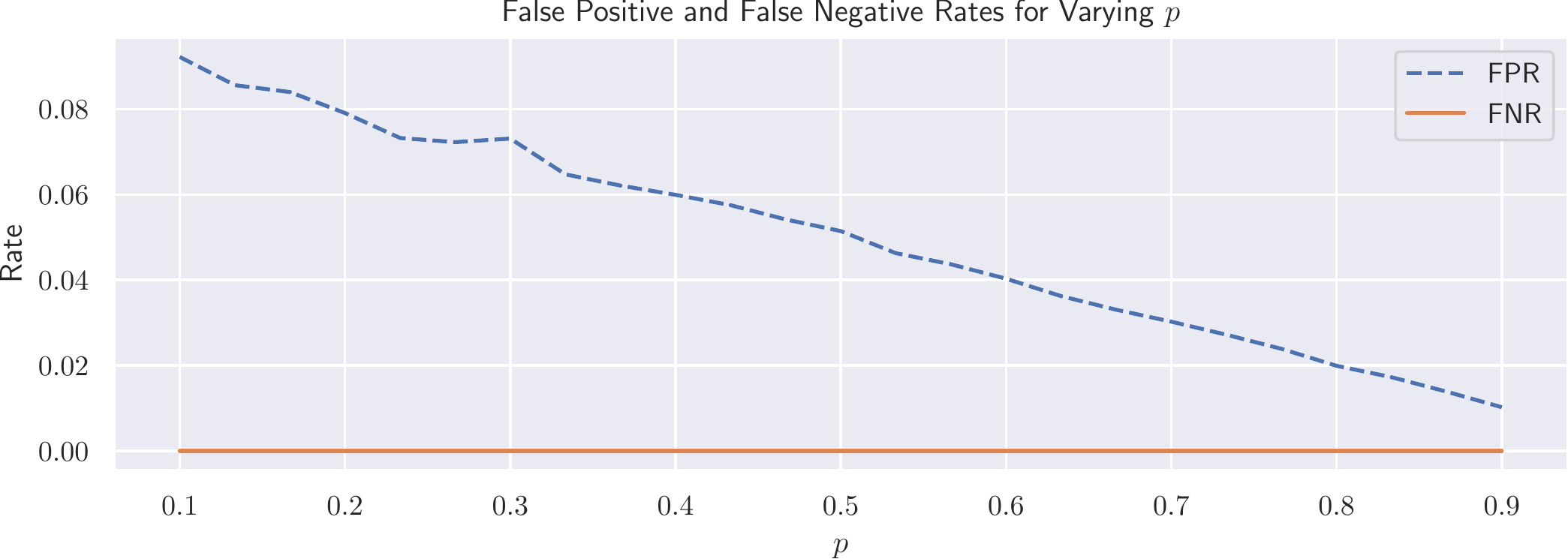}
\caption{Varying $p$ with $\sigma_{\delta} = \sqrt{2}$.}
\label{fig:fnr-fpr-in-p}
\end{subfigure}\\
\begin{subfigure}{.5\textwidth}
\centering
\includegraphics[width=3in]{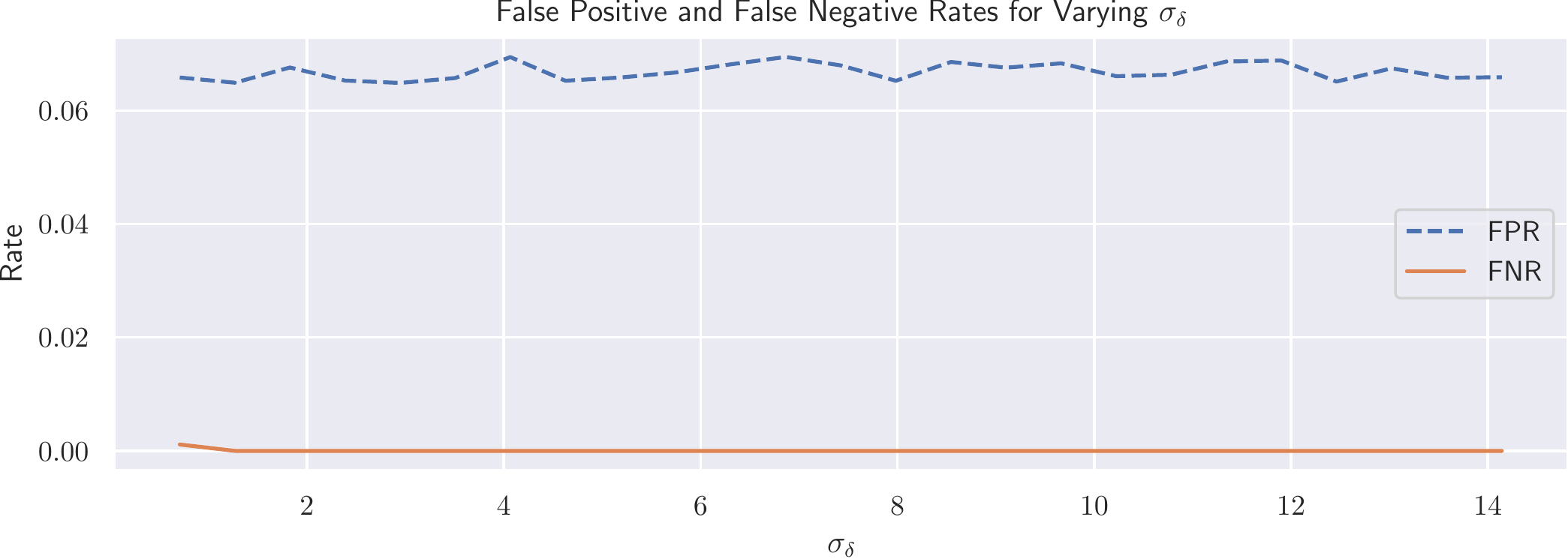}
\caption{Varying $\sigma_{\delta}$ with $p = .33$.}
\label{fig:fnr-fpr-in-sigma}
\end{subfigure}
\caption{False positive and false negative rates for the random interference experiments of section \ref{sec:randint}. (a) Fixed $\sigma_{\delta}$ and varying $p$. (b) Fixed $p$ and varying $\sigma_{\delta}$.}
\label{fig:fnr-fpr-randint}
\end{figure}

\subsection{Directed Interference with Random Amplitude} \label{sec:randamp}

In our next set of experiments, we consider the setting of interference along a fixed direction of arrival with random amplitude. We take $\delta(t) = r(t) \, |\xi(t)| \, B $, where $r(t) \sim \mathrm{Bernoulli}(p)$ and $\xi(t) \in N(0, \sigma_{\delta}^{2})$. The fixed matrix $B \in \mathbb{C}^{m \times 1}$ is the array response vector for a signal with direction of arrival $\pi/2$. Thus, this experiment captures an interference term which is nonzero with probability $p$. Nonzero interference has amplitude which takes random values according to scaled $\chi$ distribution, and fixed direction of arrival $\pi/2$. We take the noise term $\epsilon(t) \sim CN(0, \sigma_{\epsilon}^{2} \mathrm{I}_{50})$ for all $t$, with $\sigma_{\epsilon} = \sqrt{2}/2$ as in section \ref{sec:randint}.

Table \ref{tab:randamp-conmat} gives the confusion matrix for an experiment conducted with $p = 0.1$ and $\sigma_{\delta} = \sqrt{2}$. The false positive rate of $\approx 0.106$ is quite close to $0.1$, our desired bound on the false discard rate. In this experiment we also have some false negatives, which results in $\approx 4\%$ of observations estimated to not contain interference actually containing nonzero interference. Empirically, these observations are situations where the amplitude $|\xi(t)|$ of the interference term is small, so that interference is more easily confused with the noise term $\epsilon(t)$. This is clear from the fact that the impact on the direction of arrival estimate is minimal--the direction of arrival estimate when discarding the observations estimated to contain interference is $0.2506 \pi$, which is quite close to true value of $\pi/4$ and better than the $0.2597 \pi$ estimate generated from using all observations.
\begin{table}[h]
\centering
\begin{tabular}{ll|rr}
\toprule
{} & {} &  \multicolumn{2}{c}{Estimation} \\
{} & {} &  No Interference & Interference \\
\midrule
\multirow{2}{*}{Truth} &  No Interference &                    89349          &         729  \\
{} & Interference    &                    3793           &       6129  \\
\bottomrule
\end{tabular}
\caption{$\sigma_{\delta} = \sqrt{2}$, $p = 0.1$. In this example the true direction of arrival was $\pi/4$, and the estimate produced by discarding observations estimated to contain interference was $0.2506 \pi$ as compared to $0.2597 \pi$ when using all observations.}
\label{tab:randamp-conmat}
\end{table}

Figure \ref{fig:fnr-fpr-randamp} examines the false positive rate and false negative rate for a wide range of $p$ and $\sigma_{\delta}$ values. In Figure \ref{fig:fnr-fpr-in-p-randamp}, the clear trend is that increasing the proportion of observations containing interference results in an increase in both false positive and false negative rates. For $p < 0.2$, where (in expectation) less than $20\%$ of observations contain interference, the false positive rate is only slightly above our desired bound of $q=0.1$, and the false negative rate slightly below that, but $p > 0.2$ results in a rapid increase in both error rates. Unsurprisingly, detecting which observations contain interference is more difficult when interference occurs in higher proportion of observations.

For fixed $p=0.1$ and $\sigma_{\delta}$ varying, Figure \ref{fig:fnr-fpr-in-sigma-randamp} demonstrates that increasing the amplitude of the interfering signal results in a slightly increased false positive rate, while the false negative rate decreases. This is especially promising because impulse noise is often extremely sparse with heavy tails, so the fact that our estimator has only minimal loss of performance for sparse interference containing extreme values confirms its utility in those settings.

\begin{figure}
\centering
\begin{subfigure}[a]{.5\textwidth}
\centering
\includegraphics[width=3in]{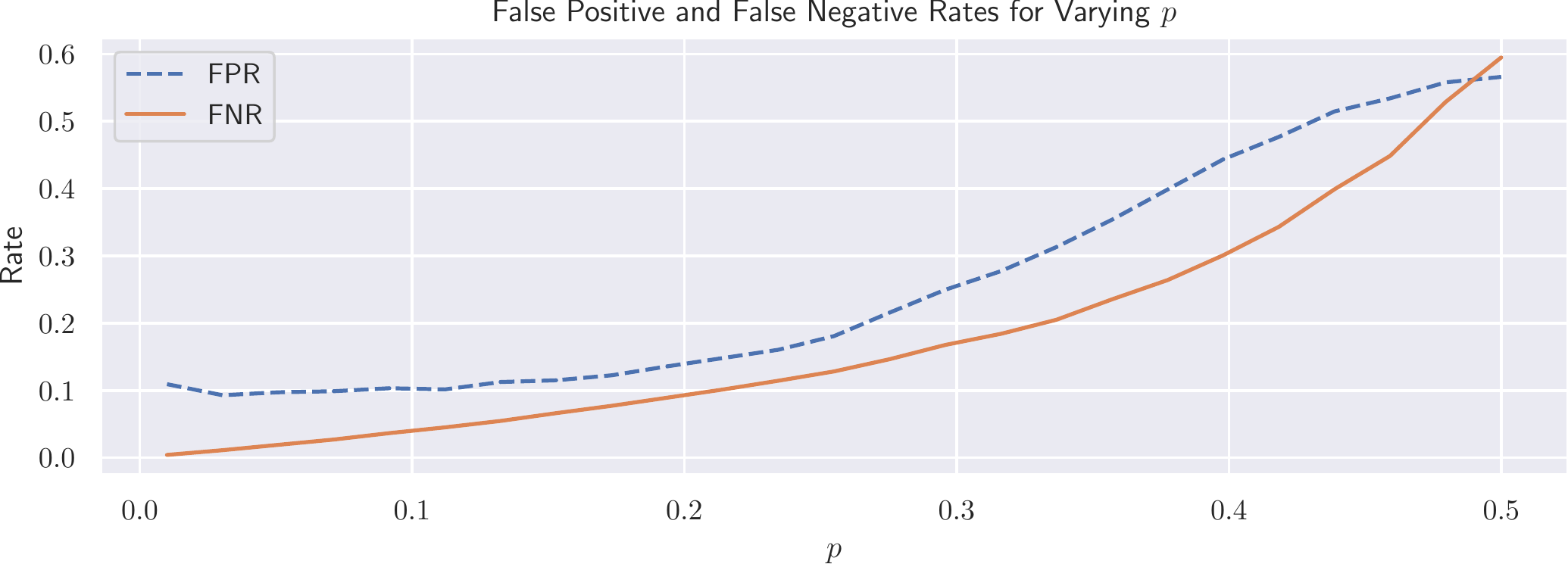}
\caption{Varying $p$ with $s_{\delta} = \sqrt{2}$.}
\label{fig:fnr-fpr-in-p-randamp}
\end{subfigure}\\
\begin{subfigure}{.5\textwidth}
\centering
\includegraphics[width=3in]{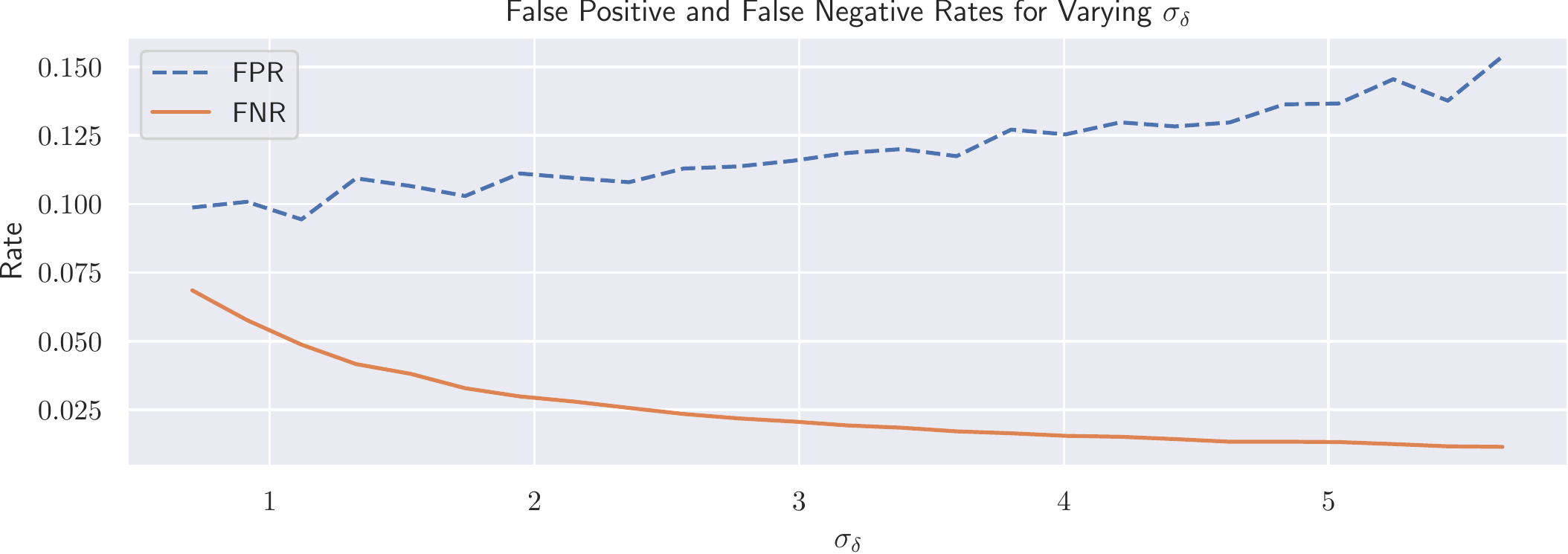}
\caption{Varying $\sigma_{\delta}$ with $p = 0.1$.}
\label{fig:fnr-fpr-in-sigma-randamp}
\end{subfigure}
\caption{False positive and false negative rates for the random amplitude experiments of section \ref{sec:randamp}. (a) Fixed $\sigma_{\delta}$ and varying $p$. (b) Fixed $p$ and varying $\sigma_{\delta}$.}
\label{fig:fnr-fpr-randamp}
\end{figure}

\subsection{Directed Interference with Constant Amplitude} \label{sec:dirintf}

For our final set of experiments, we consider directed interference with constant amplitude. That is, we take $\delta(t) = r(t) \, B \, s_{\delta}$
, where $r(t) \sim \mathrm{Bernoulli}(p)$ and $s_{\delta} \in \mathbb{R}$ is some constant. As in section \ref{sec:randamp}, the matrix $B$ is the array response vector for a signal with direction of arrival $\pi/2$ relative to the axis of the uniform linear array.

Table \ref{tab:dirintf-conmat} gives the confusion matrix for one such experiment conducted with $p = 0.1$ and $s_{\delta} = 1.0$.

\begin{table}[h]
\centering
\begin{tabular}{ll|rr}
\toprule
{} & {} &  \multicolumn{2}{c}{Estimation} \\
{} & {} &  No Interference & Interference \\
\midrule
\multirow{2}{*}{Truth} &  No Interference &                 88768       &           1310  \\
{} & Interference    &                   0      &            9922 \\
\bottomrule
\end{tabular}
\caption{$s_{\delta} = 1.0$, $p = 0.1$. In this example, the true direction of arrival was $\pi/4$, and the estimate produced by discarding observations estimated to contain interference was $0.2494\pi$, compared to $0.2622\pi$ when using all observations.}
\label{tab:dirintf-conmat}
\end{table}

Figure \ref{fig:fnr-fpr-dirintf} gives directed interference results for variable $p$ and $s_{\delta}$. In Figure \ref{fig:fnr-fpr-in-p-dirintf}, with $s_{\delta}$ fixed at $1.0$, we see that the false positive rate is near our desired bound of $q=0.1$ for $p \leq 0.1$, but that increasing the probability of nonzero interference beyond $0.1$ results in a rapid increase in the false positive rate. Unlike Figure \ref{fig:fnr-fpr-in-p-randamp}, the false positive rate increases while the false negative rate stays small. Since all errors are false positives we see a more rapid rise in false positive rate than in Figure \ref{fig:fnr-fpr-in-p-randamp}, so that at $p = 0.3$ we have $50\%$ of observations estimated to contain nonzero interference actually contain no interference.

The situation for high amplitude interference confirms the trend observed in Figures \ref{fig:fnr-fpr-in-sigma} and \ref{fig:fnr-fpr-in-sigma-randamp}. In Figure \ref{fig:fnr-fpr-in-sigma-dirintf}, we see that high powered interference results in only a small increase to the false positive rate. In particular, we see that the false positive rate only increases to $0.125$ from our desired bound of $0.1$ when the amplitude of the interference term is five times larger than the signal of interest.

\begin{figure}
\centering
\begin{subfigure}[a]{.5\textwidth}
\centering
\includegraphics[width=3in]{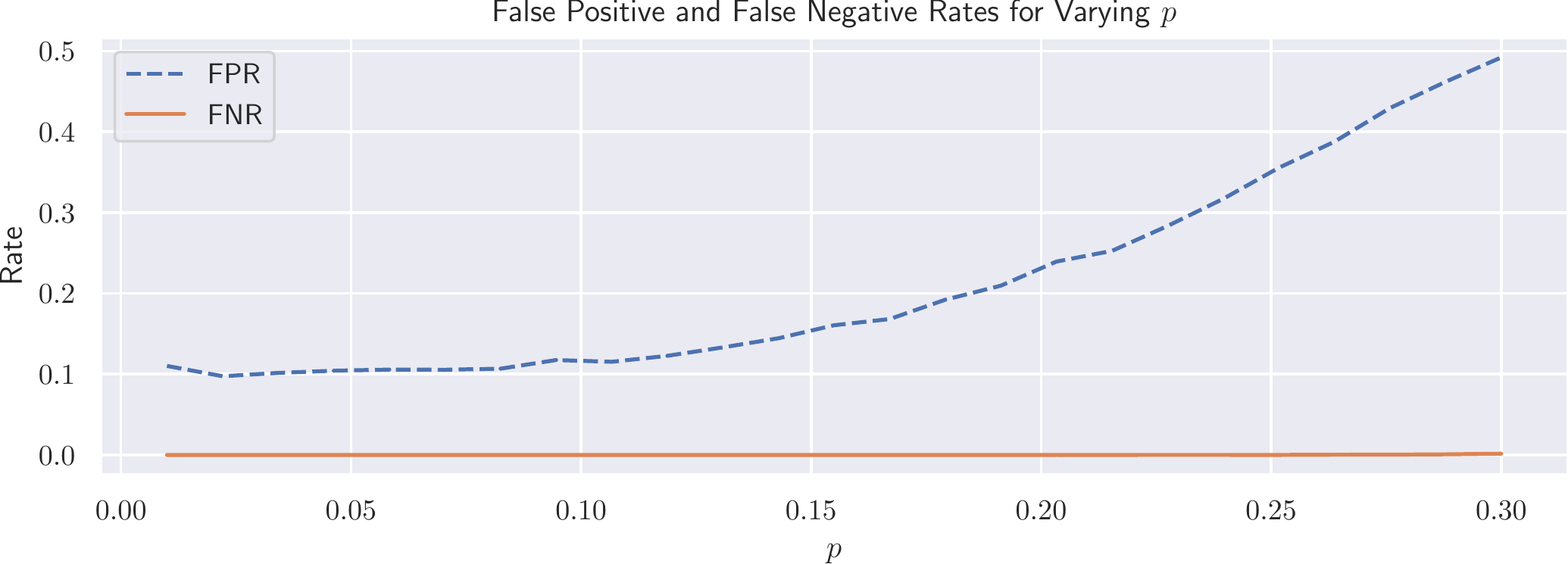}
\caption{Varying $p$ with $s_{\delta} = 1$.}
\label{fig:fnr-fpr-in-p-dirintf}
\end{subfigure}\\
\begin{subfigure}{.5\textwidth}
\centering
\includegraphics[width=3in]{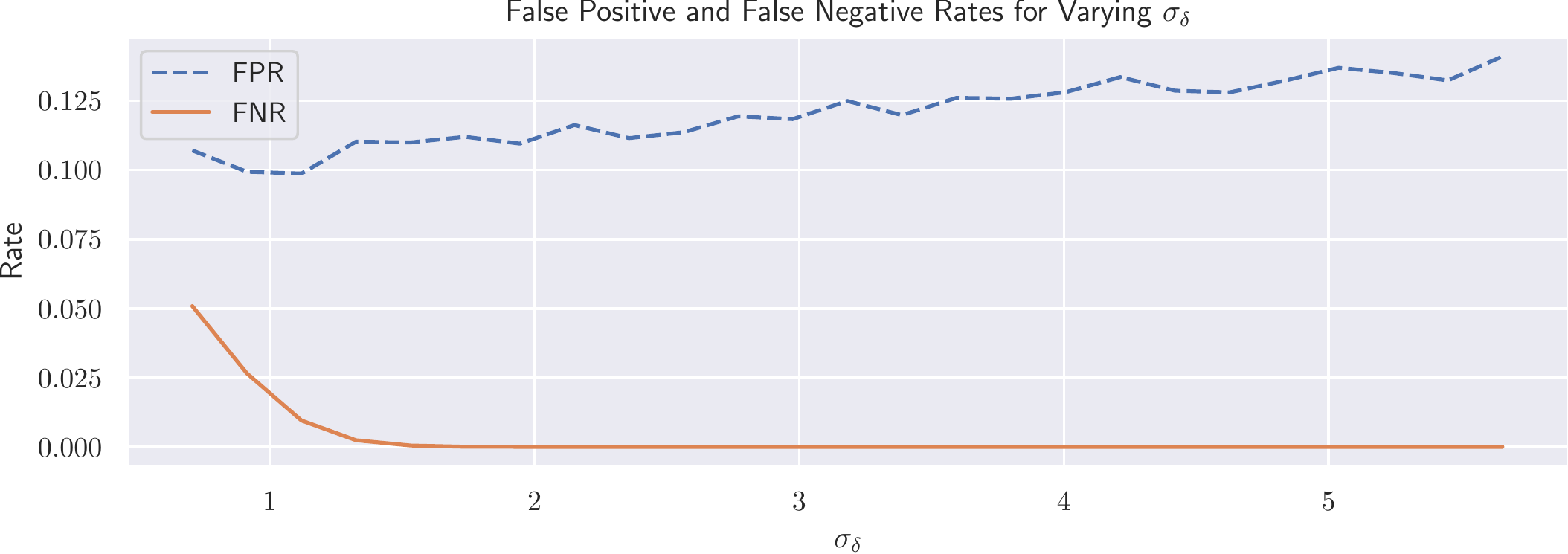}
\caption{Varying $s_{\delta}$ with $p = .05$.}
\label{fig:fnr-fpr-in-sigma-dirintf}
\end{subfigure}
\caption{False positive and false negative rates for the directed experiments of section \ref{sec:dirintf}. (a) Fixed $s_{\delta}$ and varying $p$. (b) Fixed $p$ and varying $s_{\delta}$.}
\label{fig:fnr-fpr-dirintf}
\end{figure}

\subsection{Computational Efficiency}

We conclude this section with empirical evidence of the computational efficiency of Algorithm \ref{altmin}. Figure \ref{fig:hist} displays a histogram of running times for the algorithm in the 154 experiments conducted in subsections \ref{sec:randint}-\ref{sec:dirintf}. The reported times are using a Python/Cython implementation of Algorithm \ref{altmin} on a laptop with Intel(R) Core(TM) i9-9980HK processor running at 2.40 GHz, 64.0 GB RAM, and Debian 11 operating system.

Recall that for all simulations in this section, we have assumed a sampling rate of 10 kHz and taken 100k observations corresponding to 10 seconds worth of data. This sampling rate is typical of what one might see in underwwater acoustics applications. In the context of analyzing 10 seconds worth of streaming data, the running times given in Figure \ref{fig:hist} are promising because they suggest that our proposed method for detecting interference could be used in real-time settings. Algorithm \ref{altmin} runs in less than $6$ seconds for all of the computational experiments in this section, with the vast majority completing in less than $3$ seconds. Since the running time for the algorithm is much less than the 10 seconds of data collected, our estimator is appropriate for use with streaming data. In those settings, running times could likely be further improved using hardware-specific implementations.

\begin{figure}[h]
\centering
\includegraphics[width=3in]{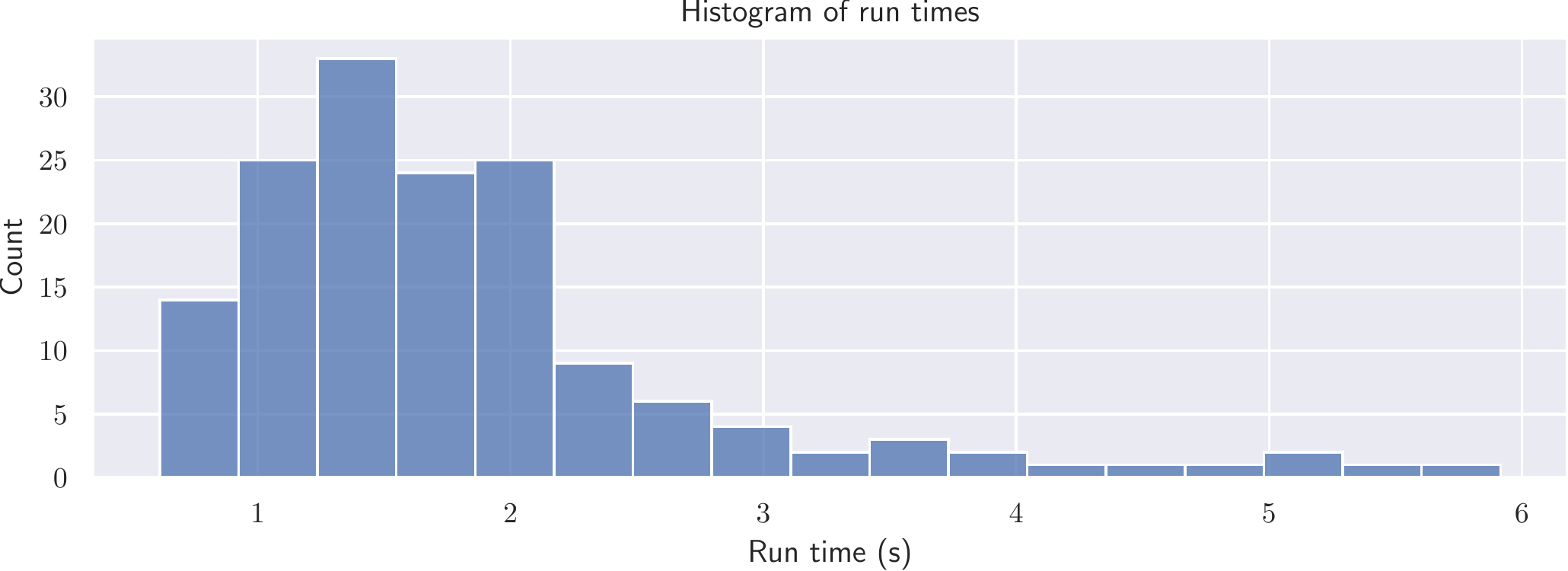}
\caption{Histogram of running times for experiments in sections \ref{sec:randint}-\ref{sec:dirintf}.}
\label{fig:hist}
\end{figure}

\section{Conclusion}\label{sec:conclusion}

In this paper, we propose a SLOPE-regularized estimator for signal subspace estimation in the presence of impulsive noise. We provide an alternating minimization algorithm for computing this estimator. We show that limit points of this algorithm exist and are local minimizers of the optimization problem which defines the estimator. Once an observation is estimated to contain nonzero interference, we propose to discard that observation and estimate the signal subspace on the remaining observations using one of a number of conventional signal subspace estimation methods. To justify the use of our estimator, we prove a result which provides finite sample control on the false discovery rate of the estimator--the expected ratio of incorrectly discarded observations to the total number of discarded observations. As opposed to previous false discovery rate results for group SLOPE norm regularized estimators, our result holds without any additional conditions on the distribution of the noise term or orthogonality design matrix conditions. We conclude with numerous simulations which test the performance of our estimator under a variety of conditions. These simulations support the efficacy of our proposed estimator for detecting and removing sparse interference in the context of signal subspace estimation.


\printbibliography

\begin{IEEEbiography}[{\includegraphics
[width=1in,height=1.25in,clip,
keepaspectratio]{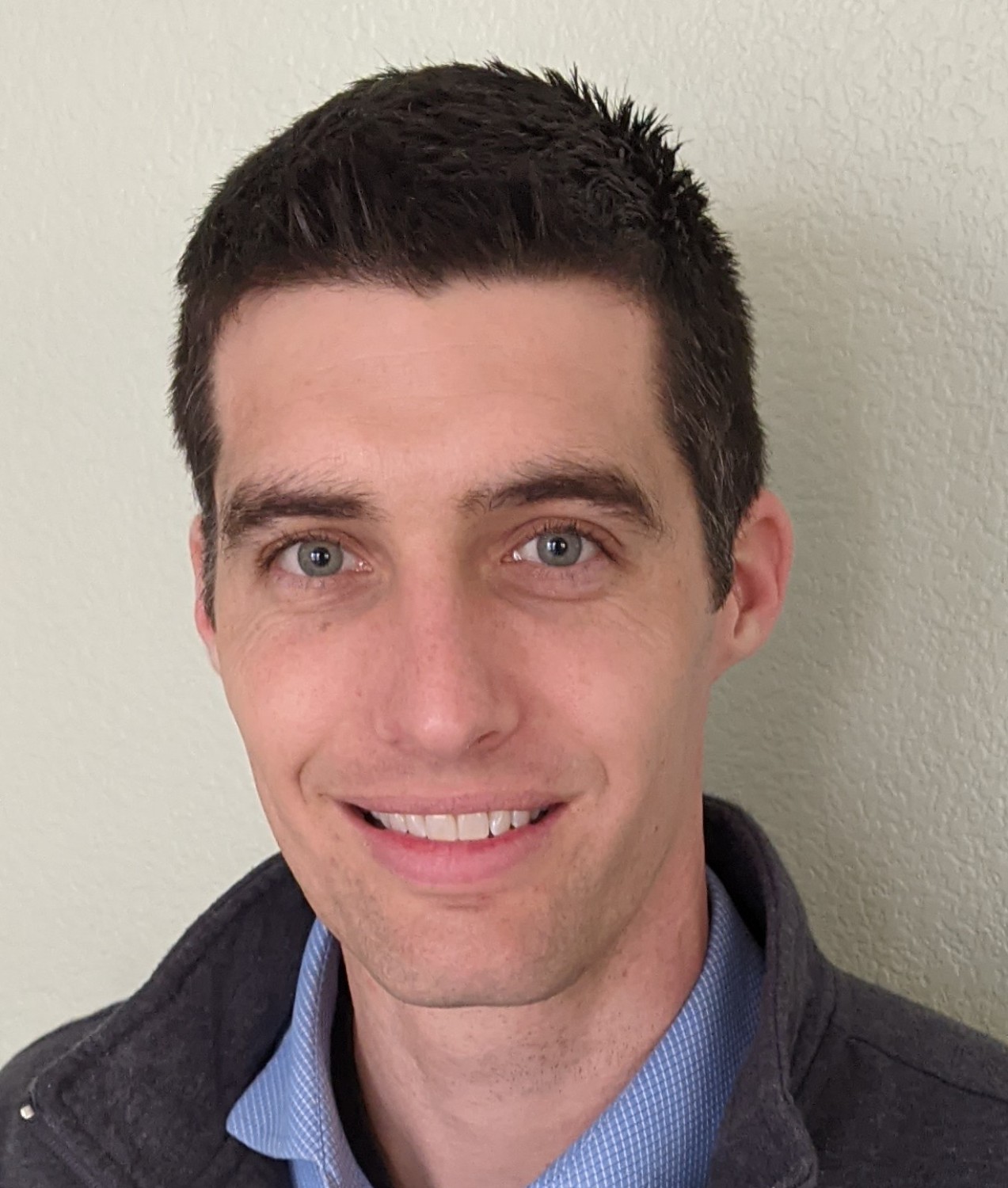}}]
{Robert L. Bassett} received the B.S. degree in mathematics, \emph{magna cum laude}, from the California State University, Bakersfield, in 2013 and the M.S. and Ph.D. degrees in mathematics from the University of California, Davis, in 2018.

He is an assistant professor of operations research at the Naval Postgraduate School. His current interests include applications of mathematical programming to problems in statistics and statistical signal processing. Dr. Bassett is primarily motivated by scientific problems relevant to issues of national security. His work has been funded by the Office of Naval Research and other agencies within the United States Department of Defense.

\end{IEEEbiography}

\begin{IEEEbiography}[{\includegraphics
[width=1in,height=1.25in,clip,
keepaspectratio]{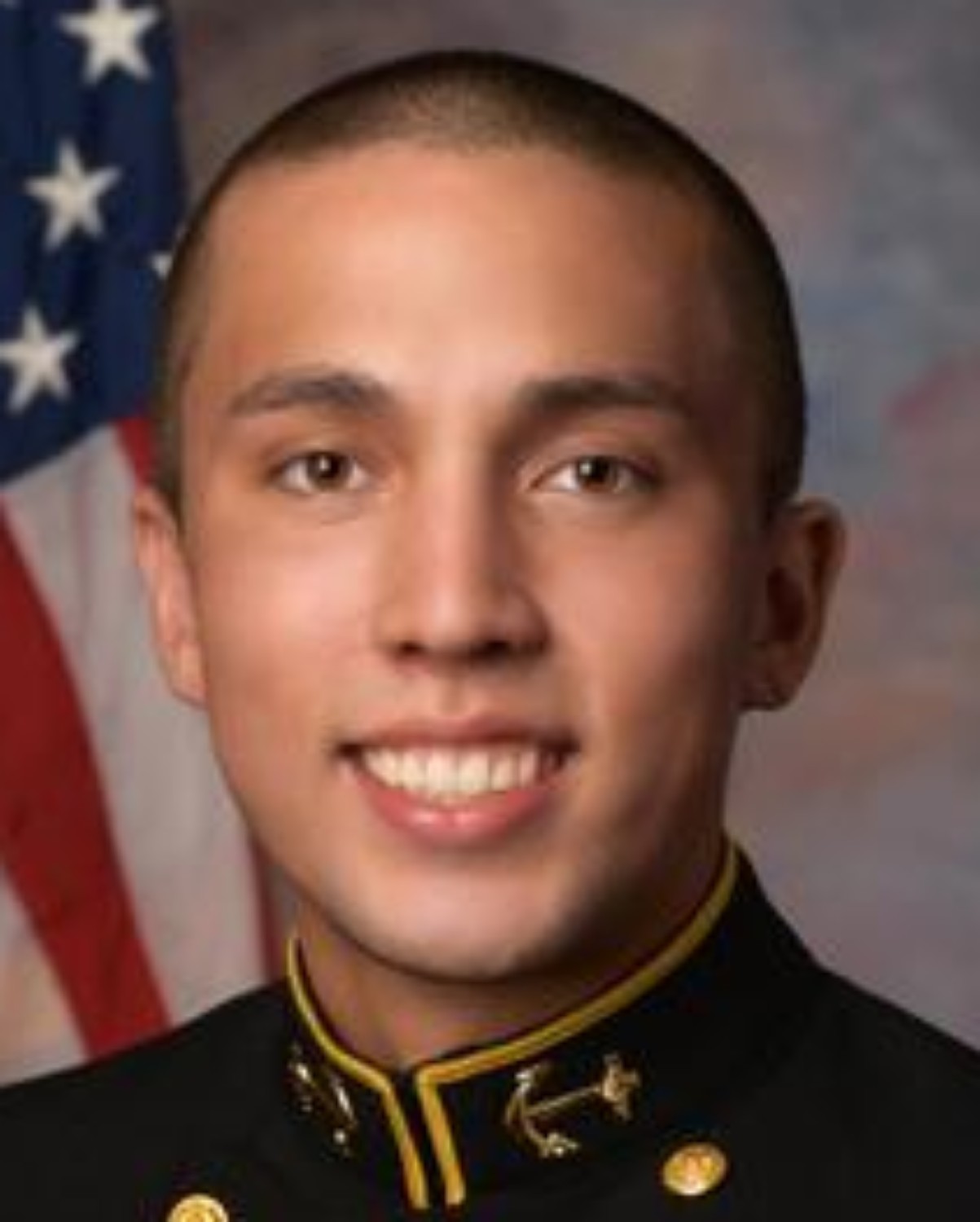}}]
{Micah Y. Oh} received the B.S. degree in applied mathematics from the United States Naval Academy in 2021 and the M.S. degree in operations research from the Naval Postgraduate School in 2022.

After graduating from the US Naval Academy in 2021, Micah Oh commissioned into his current position as an Ensign in the United States Navy. Outside of scientific pursuits, Ensign Oh is a competitive swimmer. In 2021, he was the US Naval Academy's first-ever winner of the NCAA's Elite 90 award, which is awarded to the student-athlete with the highest cumulative grade point average who competes at the NCAA championships.
\end{IEEEbiography}


\newpage

\end{document}